\documentclass[journal, onecolumn, draft]{IEEEtran}

\ifCLASSINFOpdf
\else
\fi

\usepackage[cmex10]{amsmath}

\usepackage{amssymb,amsthm}
\usepackage{color}
\usepackage{graphics}

\usepackage{multirow}
\usepackage{hhline}

\usepackage{tikz}
\usetikzlibrary{angles,calc,intersections,quotes,arrows.meta}
\usetikzlibrary{calc, arrows,backgrounds,positioning,fit}
\usetikzlibrary{decorations.markings}
\usetikzlibrary{shapes.geometric}
\usetikzlibrary{shapes.misc, fit}
\usetikzlibrary{decorations.pathreplacing}
\usepackage{pgfplots}
\tikzset{middlearrow/.style={
        decoration={markings,
            mark= at position 0.6 with {\arrow{#1}} ,
        },
        postaction={decorate}
    }
}

\allowdisplaybreaks

\usepackage[noadjust]{cite}

\makeatletter
\newtheorem*{rep@theorem}{\rep@title}
\newcommand{\newreptheorem}[2]{%
\newenvironment{rep#1}[1]{%
 \def\rep@title{#2 \ref{##1}}%
 \begin{rep@theorem}}%
 {\end{rep@theorem}}}
\makeatother

\theoremstyle{definition}
\newtheorem{definition}{Definition}
\newtheorem{construction}{Construction}
\newtheorem{example}{Example}

\theoremstyle{plain}
\newtheorem{theorem}{Theorem}
\newreptheorem{theorem}{Theorem}
\newtheorem{proposition}{Proposition}
\newtheorem{lemma}{Lemma}
\newtheorem{remark}{Remark}
\newtheorem{corollary}{Corollary}

\tikzset{%
pics/.cd,
lower/.style args={#1#2}{
  code={\node (a) {#2};
       \draw[thick,#1] (a.south west) |- (a.south east);
  }
},
}

\tikzset{%
pics/.cd,
upper/.style args={#1#2}{
  code={\node (a) {#2};
       \draw[thick,#1] (a.north west) |- (a.north east);
  }
},
}

\tikzset{%
pics/.cd,
left/.style args={#1#2}{
  code={\node (a) {#2};
       \draw[thick,#1] (a.south west) -| (a.north west);
  }
},
}

\tikzset{%
pics/.cd,
right/.style args={#1#2}{
  code={\node (a) {#2};
       \draw[thick,#1] (a.south east) -| (a.north east);
  }
},
}

\newcommand{\upmapsto}{\rotatebox[origin=c]{-270}{$\scriptstyle\mapsto$}\mkern2mu}

\DeclareMathOperator{\Rk}{Rk}
\DeclareMathOperator{\diag}{Diag}

\DeclareMathOperator{\dd}{d}
\DeclareMathOperator{\wt}{wt}

\begin{document}
%
\title{Universal and Dynamic Locally Repairable Codes with Maximal Recoverability via Sum-Rank Codes}
%
%
%
\author{Umberto~Mart{\'i}nez-Pe\~{n}as,~\IEEEmembership{Member,~IEEE,} 
		and Frank~R.~Kschischang,~\IEEEmembership{Fellow,~IEEE,}
\thanks{This work is supported by The Independent Research Fund Denmark under Grant No. DFF-7027-00053B. }
\thanks{Parts of this paper were presented at the 56th Annual Allerton Conference on Communication, Control, and Computing, Monticello, IL, USA, 2018.}
\thanks{U. Mart{\'i}nez-Pe\~{n}as and F.~R.~Kschischang are with The Edward
S. Rogers Sr. Department of Electrical and Computer Engineering, University
of Toronto, Toronto, ON M5S 3G4, Canada. (e-mail: umberto@ece.utoronto.ca;
frank@ece.utoronto.ca). }}


%
%

\markboth{}%
{}
%



\maketitle



\begin{abstract}
Locally repairable codes (LRCs) are considered with equal or unequal localities, local distances and local field sizes. An explicit two-layer architecture with a sum-rank outer code is obtained, having disjoint local groups and achieving maximal recoverability (MR) for all families of local linear codes (MDS or not) simultaneously, up to a specified maximum locality $ r $. Furthermore, the local linear codes (thus the localities, local distances and local fields) can be efficiently and dynamically modified without global recoding or changes in architecture or outer code, while preserving the MR property, easily adapting to new configurations in storage or new hot and cold data. In addition, local groups and file components can be added, removed or updated without global recoding. The construction requires global fields of size roughly $ g^r $, for $ g $ local groups and maximum or specified locality $ r $. For equal localities, these global fields are smaller than those of previous MR-LRCs when $ r \leq h $ (global parities). For unequal localities, they provide an exponential field size reduction on all previous best known MR-LRCs. For bounded localities and a large number of local groups, the global erasure-correction complexity of the given construction is comparable to that of Tamo-Barg codes or Reed-Solomon codes with local replication, while local repair is as efficient as for the Cartesian product of the local codes. Reed-Solomon codes with local replication and Cartesian products are recovered from the given construction when $ r=1 $ and $ h = 0 $, respectively. The given construction can also be adapted to provide hierarchical MR-LRCs for all types of hierarchies and parameters. Finally, subextension subcodes and sum-rank alternant codes are introduced to obtain further exponential field size reductions, at the expense of lower information rates.
\end{abstract}


%
\IEEEpeerreviewmaketitle

\section{Introduction} \label{sec intro}
%
%
%
%

\IEEEPARstart{D}{istributed} storage systems (DSSs) are of increasing importance for various cloud-based services and other applications, but are usually vulnerable to node erasures (due to disk failures). This has recently motivated several interesting and highly non-trivial coding-theoretic problems. A simple solution is data replication, but it suffers from low information rate. Optimal information rates are achieved by maximum distance separable (MDS) codes. However, repairing a single node with MDS codes requires contacting a large number of nodes and decoding all information symbols, resulting in high repair latency. Hybrid solutions are MDS codes with local replication and Cartesian products of MDS codes, but these suffer from low information rate and low global erasure-correction capability, respectively.

Among different proposals, \textit{locally repairable codes} (LRCs) \cite{pyramid, oggier, gopalan} have attracted considerable attention recently, since they allow a failed node to be repaired by contacting only a small number $ r $ (called \textit{locality}) of other nodes while simultaneously having a good global erasure-correction capability. LRCs have already been implemented in practice by Microsoft \cite{azure} and Facebook \cite{xoring}. Singleton-type bounds on the global distance of LRCs were given in \cite{gopalan, kamath}. \textit{Optimal LRCs}, meaning LRCs whose global distance attain such bounds, were obtained in \cite{gopalan, kamath, lastras, rawat, papailio-LRC, song, tamo-matroid}, and the first general construction with linear field sizes (i.e., scaling linearly with block length) was obtained in \cite{tamo-barg}. Recently, optimal LRCs with larger code lengths than the field size were obtained for certain choices of global distance, local distance and/or locality in \cite{BargAGlrc, li-lrc, luo-lrc, howlonglrc}. 

Later, LRCs where each local group has a different locality $ r_i $ (depending on the local group index $ i $) were introduced independently in \cite{kadhe} and \cite{zeh-multiple}. The main motivation is that different storage configurations may be required, or some nodes may need faster local repair or access (\textit{hot data}), while global erasure correction is improved by also considering the different non-maximum localities. Including also multiple local distances $ \delta_i \geq 2 $ was considered independently in \cite{chen-hao} and \cite{kim}. In particular, \cite{kadhe, kim} obtain optimal LRCs with multiple localities (and local distances in \cite{kim}), for arbitrary parameters, by adapting the construction from \cite{rawat} based on Gabidulin codes \cite{gabidulin, roth}, which requires field sizes that are exponential in the code length. General optimal LRCs with multiple localities, local distances and subexponential field sizes are not known yet, to the best of our knowledge. 

In another line of research, LRCs with \textit{maximal recoverability} (MR-LRCs, also known as \textit{partial MDS} or \textit{PMDS codes}) have been introduced successively in \cite{MR-property, blaum-RAID, gopalan-MR}. MR-LRCs are a strictly stronger class of LRCs than ``optimal'' LRCs, in spite of the terminology. Not only do they attain optimal global distance, but they can correct any erasure pattern that is information-theoretically correctable given the local constraints (while optimal LRCs, such as \cite{tamo-barg}, cannot). MR-LRCs over relatively small fields are significantly harder to obtain than optimal LRCs. In fact, certain parameters of MR-LRCs require super-linear field sizes \cite{gopi}, and slightly more complex topologies require super-polynomial field sizes \cite{grid-like}. Not many explicit constructions are known. Some families for certain parameters are given in \cite{blaum-RAID, gopalan-MR, blaum-twoparities, hu}. Constructions for general parameters are given in \cite{calis, gabrys, neri, guruswami-MR}. See Section \ref{sec comparisons} for a detailed description of the code parameters and field sizes achieved by the works \cite{tamo-barg, blaum-RAID, blaum-twoparities, hu, gopalan-MR, calis, gabrys, neri, guruswami-MR}. 

In this work, we propose replacing Gabidulin codes \cite{gabidulin, roth} by \textit{linearized Reed-Solomon codes} \cite{linearizedRS} in known MR-LRC constructions \cite{rawat, kadhe, kim, calis}. The key idea is that the only property of Gabidulin codes used here is that they are \textit{maximum sum-rank distance} (MSRD) block codes for the sum-rank length partition $ N = \sum_{i=1}^g r_i $ ($ = gr $ for equal localities) for $ g $ local groups, see Section \ref{sec preliminaries sum-rank}. Linearized Reed-Solomon codes are a hybrid between Gabidulin codes \cite{gabidulin, roth} and Reed-Solomon codes \cite{reed-solomon} that are MSRD and attain the minimum field-size exponent, $ r = \max_i r_i $, for the corresponding sum-rank length partition (see Proposition \ref{prop minimum m for universal LRC}). 

As a consequence, we obtain new general MR-LRCs for any choice of (equal or unequal) localities up to a specified number $ r $, arbitrary (equal or unequal) local distances, and with local field sizes of order $ \mathcal{O}(r) $. The global field size is roughly $ g^r $, for $ g $ local groups, independent of the code \textit{dimension} $ k $ or the number of \textit{global parities} $ h = gr - k $, and global erasure correction has quadratic complexity in $ gr $ over such fields. For bounded and small localities $ r $ and large $ g $, the global decoding complexity becomes comparable to that of Tamo-Barg codes \cite{tamo-barg} and Reed-Solomon codes \cite{reed-solomon} with local replication (see Section \ref{sec comparisons}). Interestingly, the latter are recovered when $ r = 1 $. Moreover, local field sizes and complexity of local repair are actually the same as those of Cartesian products of $ r $-dimensional MDS codes (see Example \ref{example intro} in Subsection \ref{subsec partitioning initial localities}), which are recovered when $ h = 0 $. Note that local repair is assumed to be more frequent, whereas global repair is reserved to catastrophic erasures. With this construction:

1) We obtain the first general MR-LRCs for arbitrary unequal localities and local distances with global field sizes that are not exponential in $ gr $, in contrast with \cite{kadhe, kim}.

2) We obtain further field size reductions on MR-LRCs compared to \cite{gabrys} (which assumes equal localities and local distances) whenever $ r \leq h $ (see Subsection \ref{subsec other LRC}). Both small $ r $ and $ h $ are desirable in DSSs applications. Which regime, $ r \leq h $ or $ h < r $, is more desirable depends on the particular application. Observe however that large $ r $ defeats the purpose of LRCs, and $ h $ is the extra number of correctable erasures compared to the Cartesian product of the local codes (the case $ h = 0 $), hence is expected to grow somehow as $ gr $ grows. See Example \ref{example second intro} in Subsection \ref{subsec other LRC}. 

3) In contrast with most LRCs (e.g., \cite{gopalan, kamath, lastras, papailio-LRC, song, tamo-matroid, tamo-barg, zeh-multiple, chen-hao, gopalan-MR, blaum-RAID, blaum-twoparities, hu, gabrys, neri, guruswami-MR}), our construction is a) \textit{Universal}: The same architecture and outer code admits any family of $ g $ local linear codes (MDS or not) of dimensions up to $ r $; and b) \textit{Dynamic}: Arbitrary changes of local linear codes are possible, always preserving the MR condition and optimal global distance, without global recoding or changes in architecure or outer code; one simply needs to perform efficient local recodings, usually over the much smaller local fields. See Section \ref{sec dynamic properties} for details. Although MR-LRCs with $ \delta > 2 $ (e.g., \cite{blaum-RAID, gopalan-MR, gabrys}) admit puncturing, hence changes in local distances (this seems not to be possible without the MR condition), changing localities without global recoding seems difficult. In addition, our construction allows file components and local groups to be added, removed or updated without global recoding. Universality and dynamism are of interest in DSSs, where one may want to adapt to new configurations over time without recoding all of the stored data. Universality and dynamism are shared by the particular cases $ r = 1 $ (Reed-Solomon codes with local replication) and $ h = 0 $ (Cartesian products).

4) The universality in Item 3 implies that our construction admits any combination and any number of layers of (equal or unequal) \textit{multi-layer or hierarchical localities}. This means that the local codes may be in turn MR-LRCs, and in such a case, each node is protected by a small local code (\textit{lower level}), and simultaneously by a larger local MR-LRC (\textit{medium level}) and by the global MR-LRC (\textit{higher level}), in the two-level case, and similarly for more levels. Lower-level codes may repair fewer erasures, but require contacting a smaller number of nodes and typically use a smaller field size. Hierarchical locality was introduced in \cite{hierarchical}. The concept of MR-LRCs with hierarchical localities has been itnroduced parallel to the present work in \cite{nair}. The comparison between the field sizes between the constructions in \cite{nair} and this work is analogous to the case of simple MR-LRCs as in Item 2. Observe that the work \cite{nair} considers only two-level hierarchies and equal localities and local distances for each level.  

5) We show that universal MR-LRCs as in Item 3 with localities $ r_i $ up to a specified number $ r $ are equivalent to global codes with the given architecture plus MSRD outer codes with sum-rank length partition $ N = \sum_{i=1}^g r_i $. In particular, $ r = \max_i r_i $ is the smallest field extension that allows this type of universality. Hence MR-LRCs over smaller fields always need to coordinate their global and local codes in some way, reducing universality and dynamism.

6) Following the idea in the previous item, we introduce \textit{subextension subcodes} and \textit{sum-rank alternant codes}. As in the Hamming-metric case ($ r = 1 $), sum-rank alternant codes enable exponential field size reductions with the same global erasure-correction capability, at the expense of lower information rates.

The remainder of this paper is organized as follows. In Section \ref{sec preliminaries sum-rank}, we give some preliminaries on sum-rank codes, including some new results. In Section \ref{sec main construction}, we describe our main MR-LRC construction. In Section \ref{sec universal lrc arbitrary local codes}, we study MR-LRCs where local codes can be arbitrary linear codes over some (local) subfields, including the proof that the given construction is MR. We show in Section \ref{sec dynamic properties} how to perform local linear recodings, partition localities, obtain multi-layer or hierarchical MR-LRCs, and update localities, file components and number of local groups. In Section \ref{sec comparisons}, we compare the achieved global and local fields and decoding complexities of the proposed construction with LRCs from the literature that cover general parameters. In Section \ref{sec further field reductions}, we introduce subextension subcodes and sum-rank alternant codes to obtain similar LRCs, which allows us to obtain exponential field size reductions at the expense of reducing information rates. Section \ref{sec conclusion} concludes the paper.

\subsection*{Notation}

For a field $ \mathbb{F} $, we denote by $ \mathbb{F}^{m \times n} $ the set of $ m \times n $ matrices with entries in $ \mathbb{F} $, and we denote $ \mathbb{F}^n = \mathbb{F}^{1 \times n} $. For a prime power $ q $, we denote by $ \mathbb{F}_q $ the finite field with $ q $ elements.

For a positive integer, we denote $ [n] = \{ 1,2, \ldots, n\} $. Given $ \mathcal{R} \subseteq [n] $, we denote by $ \mathbf{c}_\mathcal{R} \in \mathbb{F}^{|\mathcal{R}|} $, $ A|_\mathcal{R} \in \mathbb{F}^{m \times | \mathcal{R} |} $ and $ \mathcal{C}_\mathcal{R} \subseteq \mathbb{F}^{| \mathcal{R} |} $ the restrictions of a vector $ \mathbf{c} \in \mathbb{F}^n $, a matrix $ A \in \mathbb{F}^{m \times n} $ and a code $ \mathcal{C} \subseteq \mathbb{F}^n $, respectively, to the coordinates indexed by $ \mathcal{R} $.

In general, the term \textit{complexity $ \mathcal{O}(N) $} means complexity of $ \mathcal{O}(N) $ operations over the corresponding field.

\section{Preliminaries on Sum-rank Codes} \label{sec preliminaries sum-rank}

The sum-rank metric was introduced in \cite{multishot} for error-correction in multishot network coding. It was implicitly considered earlier in the space-time coding literature (see \cite[Sec. III]{space-time-kumar}). In Subsection \ref{subsec sum-rank codes}, we collect basic properties of sum-rank codes, including several new results. In Subsection \ref{subsec lin RS codes}, we review the construction of \textit{linearized Reed-Solomon codes} \cite{linearizedRS}, which is the only known general family of \textit{maximum sum-rank distance} (MSRD) block codes with subexponential field sizes in the code length.

\subsection{Sum-rank Codes} \label{subsec sum-rank codes}

Let $ q $ denote a prime power and fix a positive integer $ m $. Fix an ordered basis $ \mathcal{A} = \{ \alpha_1, \alpha_2, \ldots, \alpha_m \} $ of $ \mathbb{F}_{q^m} $ over $ \mathbb{F}_q $. For any non-negative integer $ s $, we denote by $ M_\mathcal{A} :
\mathbb{F}_{q^m}^s \longrightarrow \mathbb{F}_q^{m \times s} $ the corresponding \textit{matrix representation} map, given by 
\begin{equation}
M_\mathcal{A} \left( \sum_{i=1}^m \alpha_i \mathbf{c}_i \right) = \left( \begin{array}{cccc}
c_{11} & c_{12} & \ldots & c_{1s} \\
c_{21} & c_{22} & \ldots & c_{2s} \\
\vdots & \vdots & \ddots & \vdots \\
c_{m1} & c_{m2} & \ldots & c_{ms} \\
\end{array} \right),
\label{eq def matrix representation map}
\end{equation}
where $ \mathbf{c}_i = (c_{i,1}, c_{i,2}, \ldots, c_{i,s}) \in \mathbb{F}_q^s $, for $ i = 1,2, \ldots, m $. 

Fix positive integers $ g $ and $ N = r_1 + r_2 + \cdots + r_g $. The integer $ g $ will be called the \textit{initial number of local groups}, and $ r_1, r_2, \ldots, r_g $, \textit{initial localities} (not necessarily equal).

\begin{definition} [\textbf{Sum-rank metric \cite{multishot}}]
Let $ \mathbf{c} = (\mathbf{c}^{(1)}, $ $ \mathbf{c}^{(2)}, $ $ \ldots,
$ $ \mathbf{c}^{(g)}) \in \mathbb{F}_{q^m}^N $, where $
\mathbf{c}^{(i)} \in \mathbb{F}_{q^m}^{r_i} $, for $ i = 1,2, \ldots,
g $. We define the sum-rank weight of $ \mathbf{c} $ as
$$ \wt_{SR}(\mathbf{c}) = \sum_{i=1}^g {\rm
Rk}(M_{\mathcal{A}}(\mathbf{c}^{(i)})). $$

Finally, we define the sum-rank metric $ \dd_{SR} : (\mathbb{F}_{q^m}^N)^2 \longrightarrow \mathbb{N} $ as $ \dd_{SR}(\mathbf{c}, \mathbf{d}) = {\rm wt}_{SR}(\mathbf{c} - \mathbf{d}) $, for all $ \mathbf{c}, \mathbf{d} \in \mathbb{F}_{q^m}^N $. We will also say that $ N = r_1 + r_2 + \cdots + r_g $ is a \textit{sum-rank length partition}. A sum-rank length partition is thus the same as a number of initial local groups and initial localities.
\end{definition}

As usual, for a code $ \mathcal{C} \subseteq \mathbb{F}_{q^m}^N $ (linear or non-linear), we define its minimum sum-rank distance as
\begin{equation}
{\rm d}_{SR}(\mathcal{C}) = \min \{ {\rm d}_{SR}(\mathbf{c}, \mathbf{d}) \mid \mathbf{c}, \mathbf{d} \in \mathcal{C}, \mathbf{c} \neq \mathbf{d} \}. 
\label{eq def min sum-rank distance}
\end{equation}

Observe that the Hamming metric \cite{hamming} and the rank metric \cite{delsartebilinear, gabidulin, roth} are recovered from the sum-rank metric by setting $ r_1 = r_2 = \ldots = r_g = 1 $ and $ g = 1 $, respectively.

The crucial fact about the minimum sum-rank distance for (universal) global erasure correction in LRCs is that it gives the worst-case erasure-correction capability after any possible local linear recoding on disjoint local groups. This is given by the following result, which we will use throughout the paper and is of interest in its own right.

\begin{theorem} \label{th sum-rank distance is min among hamming distances}
Given a code $ \mathcal{C} \subseteq \mathbb{F}_{q^m}^N $ (linear or non-linear), it holds that
\begin{equation}
\begin{split}
{\rm d}_{SR}(\mathcal{C}) = \min \{ & {\rm d}_H(\mathcal{C} A) \mid A = \diag(A_1, A_2, \ldots, A_g), \\
& A_i \in \mathbb{F}_q^{r_i \times r_i} \textrm{ invertible}, 1 \leq i \leq g \}.
\end{split}
\end{equation}
Here, $ d_H(\mathcal{C}A) $ denotes the minimum Hamming distance of the code $ \mathcal{C}A \subseteq \mathbb{F}_{q^m}^N $, where the Hamming distance between two codewords $ \mathbf{c}, \mathbf{d} \in \mathbb{F}_{q^m}^N $ is defined as $ {\rm d}_H(\mathbf{c}, \mathbf{d}) = {\rm wt}_H(\mathbf{c} - \mathbf{d}) $, where
$$ {\rm wt}_H(\mathbf{e}) = |\{ i \in [N] \mid e_i \neq 0 \}| , $$
for any vector $ \mathbf{e} = (e_1, e_2, \ldots, e_N) \in \mathbb{F}_{q^m}^N $, where $ e_i \in \mathbb{F}_{q^m} $, for $ i = 1,2, \ldots, N $.  
\end{theorem}
\begin{proof}
We first prove the inequality $ \leq $. Since multiplying by such block-diagonal matrices $ A $ constitutes a linear sum-rank isometry, and sum-rank distances are upper bounded by Hamming distances, we deduce that
$$ {\rm d}_{SR}(\mathcal{C}) = {\rm d}_{SR}(\mathcal{C} A) \leq {\rm d}_H(\mathcal{C}A), $$
for all such matrices, and the inequality follows. 

We now prove the inequality $ \geq $. Let $ \mathbf{c}, \mathbf{d} \in \mathcal{C} $ be such that $ \mathbf{c} \neq \mathbf{d} $ and $ {\rm d}_{SR} (\mathcal{C}) = {\rm d}_{SR}(\mathbf{c}, \mathbf{d}) $. Let $ \mathbf{c} = (\mathbf{c}^{(1)}, \mathbf{c}^{(2)}, \ldots, \mathbf{c}^{(g)}) $ and $ \mathbf{d} = (\mathbf{d}^{(1)}, \mathbf{d}^{(2)}, \ldots, \mathbf{d}^{(g)}) $, where $ \mathbf{c}^{(i)}, \mathbf{d}^{(i)} \in \mathbb{F}_{q^m}^{r_i} $, for $ i = 1,2, \ldots, g $. By column reduction, there exists an invertible matrix $ A_i \in \mathbb{F}_q^{r_i \times r_i} $ such that
\begin{equation*}
\begin{split}
M_\mathcal{A}((\mathbf{c}^{(i)} - \mathbf{d}^{(i)}) A_i) & = M_\mathcal{A}(\mathbf{c}^{(i)} - \mathbf{d}^{(i)}) A_i \\
& = (B_i , 0_{r_i - w_i}) \in \mathbb{F}_q^{m \times r_i},
\end{split}
\end{equation*}
for certain full-rank matrix $ B_i \in \mathbb{F}_q^{m \times w_i} $, where $ w_i = \Rk ( M_\mathcal{A}(\mathbf{c}^{(i)} - \mathbf{d}^{(i)})) $, for $ i = 1,2, \ldots, g $. In particular, we deduce that $ \wt_H ((\mathbf{c}^{(i)} - \mathbf{d}^{(i)})A_i) = w_i $, for $ i = 1,2, \ldots, g $.

Define $ A = \diag(A_1, A_2, \ldots, A_g) \in \mathbb{F}_q^{N \times N} $. It follows that $ {\rm wt}_{SR}((\mathbf{c} - \mathbf{d})A) = \sum_{i=1}^g w_i = {\rm w}_H((\mathbf{c} - \mathbf{d})A) $. Hence 
\begin{equation*}
\begin{split}
{\rm d}_{SR}(\mathcal{C}) = {\rm d}_{SR}(\mathcal{C}A) & = {\rm d}_{SR}(\mathbf{c}A, \mathbf{d}A) = {\rm w}_{SR}((\mathbf{c} - \mathbf{d})A) \\
 & = {\rm w}_H((\mathbf{c} - \mathbf{d})A) = {\rm d}_H(\mathbf{c}A, \mathbf{d}A) \\
 & \geq {\rm d}_H(\mathcal{C} A),
\end{split}
\end{equation*}
and the inequality is proven.
\end{proof}

The following will be the main tool for global erasure-correction of locally repairable codes based on sum-rank codes. It follows from Theorem \ref{th sum-rank distance is min among hamming distances}. It may also be deduced from \cite[Th. 1]{secure-multishot}.

\begin{corollary}[\textbf{Erasure correction}] \label{cor sum-rank erasure correction}
Let $ \mathcal{C} \subseteq \mathbb{F}_{q^m}^N $ be a code (linear or non-linear), and let $ 0 \leq \rho < N $. The following are equivalent:
\begin{enumerate}
\item
$ \rho < {\rm d}_{SR}(\mathcal{C}) $.
\item
For all integers $ n_i \geq 1 $ and all matrices $ A_i \in \mathbb{F}_q^{r_i \times n_i} $, for $ i = 1,2, \ldots, g $, such that
$$ N - \sum_{i=1}^g {\rm Rk}(A_i) \leq \rho, $$ 
there exists a decoder 
$$ D : \mathcal{C} \diag(A_1, A_2, \ldots, A_g) \longrightarrow \mathcal{C} $$ 
(depending on the matrices $ A_1 $, $ A_2 $, \ldots, $ A_g $), such that $ D(\mathbf{c} \diag(A_1, A_2, \ldots, A_g)) = \mathbf{c} $, for all $ \mathbf{c} \in \mathcal{C} $.
\end{enumerate}
\end{corollary}

From Theorem \ref{th sum-rank distance is min among hamming distances} and the Hamming-metric Singleton bound \cite{singleton}, we also obtain the following result. It may also be deduced from \cite[Th. 5]{secure-multishot}.

\begin{corollary} [\textbf{First Singleton bound}] \label{cor sum rank singleton}
Let $ \mathcal{C} \subseteq \mathbb{F}_{q^m}^N $ be a (linear or non-linear) code. It holds that
\begin{equation}
| \mathcal{C} | \leq q^{m(N - {\rm d}_{SR}(\mathcal{C}) + 1)}.
\label{eq sum-rank singleton bound}
\end{equation}
Furthermore, equality holds if, and only if, $ \mathcal{C} A \subseteq \mathbb{F}_{q^m}^N $ is MDS, for all $ A = \diag(A_1, A_2, \ldots, A_g) \in \mathbb{F}_q^{N \times N} $, such that $ A_i \in \mathbb{F}_q^{r_i \times r_i} $ is invertible, for $ i = 1,2, \ldots, g $.
\end{corollary}

A code satisfying equality in (\ref{eq sum-rank singleton bound}) is called \textit{maximum sum-rank distance} (MSRD). 

We now show that, when the sublengths are equal, $ m = N/g $ is the smallest possible extension degree of $ \mathbb{F}_{q^m} $ over $ \mathbb{F}_q $ for the existence of MSRD codes.

\begin{corollary} [\textbf{Second Singleton bound}] \label{cor minimum m for MSRD}
Assuming that $ r_1 = r_2 = \ldots = r_g = N/g $, then any (linear or non-linear) code $ \mathcal{C} \subseteq \mathbb{F}_{q^m}^N $ satisfies the bound
\begin{equation}
|\mathcal{C}| \leq \left( q^{N/g} \right)^{gm - {\rm d}_{SR}(\mathcal{C}) + 1}.
\label{eq alternative singleton bound}
\end{equation}
In particular, there exists an MSRD code $ \mathcal{C} \subsetneqq \mathbb{F}_{q^m}^N $ with $ \dd_{SR}(\mathcal{C}) > 1 $ over $ \mathbb{F}_q $ only if $ m \geq N/g $.
\end{corollary}
\begin{proof}
Let $ \mathcal{C} \subseteq \mathbb{F}_{q^m}^N = (\mathbb{F}_{q^m}^{N/g})^g $ be an arbitrary code. Define $ \mathcal{C}^T $ as the code obtained by transposing the matrix representation (\ref{eq def matrix representation map}) of each block of $ N/g $ coordinates, for each codeword in $ \mathcal{C} $. We may consider that the code $ \mathcal{C}^T $ lies in $ (\mathbb{F}_{q^{N/g}}^m)^g = \mathbb{F}_{q^{N/g}}^{gm} $. Therefore, it follows from (\ref{eq sum-rank singleton bound}) that
$$ |\mathcal{C}| = |\mathcal{C}^T| \leq \left( q^{N/g} \right)^{gm - {\rm d}_{SR}(\mathcal{C}^T) + 1}. $$
Since $ d = {\rm d}_{SR}(\mathcal{C}^T) = {\rm d}_{SR}(\mathcal{C}) $, the bound (\ref{eq alternative singleton bound}) follows.

Finally, if $ m < N/g $ and $ d > 1 $, then we have that
\begin{equation*}
\begin{split}
\left( q^{N/g} \right)^{gm - d + 1} & = q^{mN - (d-1)N/g} \\
& < q^{mN - (d-1)m} = q^{m(N - d + 1)},
\end{split}
\end{equation*}
and the code $ \mathcal{C} $ cannot attain (\ref{eq sum-rank singleton bound}), hence cannot be MSRD.
\end{proof} 

As we will see in the next subsection, linearized Reed-Solomon codes \cite{linearizedRS} achieve this minimum extension degree.

As shown later in Theorem \ref{th sum-rank partitions}, a maximum rank distance (MRD) code in $ \mathbb{F}_{q^m}^N $, such as a Gabidulin code \cite{gabidulin, roth}, is also MSRD for \textit{any} sum-rank length partition $ N = r_1 + r_2 + \cdots + r_g $. However, by taking $ g=1 $ in the previous corollary, MRD codes can only exist if $ m \geq N $. For this reason, the use of linearized Reed-Solomon codes will imply a reduction in field sizes on Gabidulin-based LRCs \cite{rawat, kadhe, kim, calis}. See also Subsection \ref{subsec smalles field initial localities}.

\subsection{Linearized Reed-Solomon Codes} \label{subsec lin RS codes}

In this subsection, we review the construction of linearized Reed-Solomon codes from \cite{linearizedRS} (see also \cite[Sec. IV]{secure-multishot}). 

Assume that $ 1 \leq g \leq q-1 $ and $ 1 \leq r_i \leq m $, for $ i = 1,2, \ldots , g $. Therefore $ N \leq (q-1)m $. Let $ \sigma : \mathbb{F}_{q^m} \longrightarrow \mathbb{F}_{q^m} $ be given by $ \sigma(a) = a^q $, for all $ a \in \mathbb{F}_{q^m} $. We need to define linear operators as in \cite[Def.~20]{linearizedRS}. 

\begin{definition}[\textbf{Linear operators \cite{linearizedRS}}] \label{def linearized operators}
Fix $ a \in \mathbb{F}_{q^m} $, and define its $ i $th norm as $ N_i(a)
= \sigma^{i-1}(a) \cdots \sigma(a)a $ for $ i \in \mathbb{N} $. Now define the $ \mathbb{F}_q
$-linear operator $ \mathcal{D}_a^i : \mathbb{F}_{q^m} \longrightarrow
\mathbb{F}_{q^m} $ by
\begin{equation}
\mathcal{D}_a^i(b) = \sigma^i(b) N_i(a) ,
\label{eq definition linear operator}
\end{equation}
for all $ b \in \mathbb{F}_{q^m} $, and all $ i \in \mathbb{N} $. Define also $ \mathcal{D}_a = \mathcal{D}_a^1 $ and observe that $ \mathcal{D}_a^{i+1} = \mathcal{D}_a \circ \mathcal{D}_a^i $, for $ i \in \mathbb{N} $. 
\end{definition}

We say that $ a,b \in \mathbb{F}_{q^m} $ are \textit{conjugate} if there exists $ c \in \mathbb{F}_{q^m}^* $ such that $ b = \sigma(c)c^{-1} a $. See \cite{lam} and \cite[Eq.~(2.5)]{lam-leroy}. Take now a primitive element $ \gamma $ of $ \mathbb{F}_{q^m} $, and note that
$$ \gamma^j \neq \sigma(c)c^{-1} \gamma^i, $$ 
for all $ c \in \mathbb{F}_{q^m}^* $ and all $ 0 \leq i < j \leq q-2 $. Hence $ \gamma^0, $ $ \gamma^1, $ $ \gamma^2, $ $ \ldots , $ $ \gamma^{q-2} $ constitute the representatives of all non-trivial disjoint conjugacy classes. Finally, take a basis $ \mathcal{B} = \{ \beta_1, \beta_2, \ldots, \beta_m \} $ of $ \mathbb{F}_{q^m} $ over $ \mathbb{F}_q $, and define the matrices
\begin{equation*}
D_i = \left( \begin{array}{cccc}
\beta_1 & \beta_2 & \ldots & \beta_{r_i} \\
\mathcal{D}_{\gamma^{i-1}} \left( \beta_1 \right) & \mathcal{D}_{\gamma^{i-1}} \left( \beta_2 \right) & \ldots & \mathcal{D}_{\gamma^{i-1}} \left( \beta_{r_i} \right) \\
\mathcal{D}_{\gamma^{i-1}}^2 \left( \beta_1 \right) & \mathcal{D}_{\gamma^{i-1}}^2 \left( \beta_2 \right) & \ldots & \mathcal{D}_{\gamma^{i-1}}^2 \left( \beta_{r_i} \right) \\
\vdots & \vdots & \ddots & \vdots \\
\mathcal{D}_{\gamma^{i-1}}^{k-1} \left( \beta_1 \right) & \mathcal{D}_{\gamma^{i-1}}^{k-1} \left( \beta_2 \right) & \ldots & \mathcal{D}_{\gamma^{i-1}}^{k-1} \left( \beta_{r_i} \right) \\
\end{array} \right),
\end{equation*}
for $ i = 1,2, \ldots, g $. The following definition is a particular
case of \cite[Def.~31]{linearizedRS}.

\begin{definition} [\textbf{Linearized Reed-Solomon codes \cite{linearizedRS}}] \label{def lin RS codes}
We define the linearized Reed-Solomon code of dimension $ k $, primitive element $ \gamma $ and basis $ \mathcal{B} $, as the linear code $
\mathcal{C}^\sigma_{L,k}(\mathcal{B},\gamma) \subseteq \mathbb{F}_{q^m}^N $ with generator matrix given by 
\begin{equation}
D = (D_1 | D_2 | \ldots | D_g) \in \mathbb{F}_{q^m}^{k \times N}.
\label{eq gen matrix lin RS codes}
\end{equation}
\end{definition}

The following result is \cite[Th. 4]{linearizedRS}.

\begin{proposition}[\textbf{\cite{linearizedRS}}] \label{prop linRS are MSRD}
The linearized Reed-Solomon code $ \mathcal{C}^\sigma_{L,k}(\mathcal{B},\gamma) \subseteq \mathbb{F}_{q^m}^N $ in Definition \ref{def lin RS codes} is a $ k $-dimensional linear MSRD code for the sum-rank length partition $ N = r_1 + r_2 + \cdots + r_g $. That is, $ \dd_{SR}(\mathcal{C}^\sigma_{L,k}(\mathcal{B},\gamma)) = N - k + 1 $.
\end{proposition}

Observe that $ m \geq r = \max_i r_i $. Therefore linearized Reed-Solomon codes achieve the minimum extension degree over $ \mathbb{F}_q $ for equal localities by Corollary \ref{cor minimum m for MSRD}. See also Proposition \ref{prop minimum m for universal LRC}.

As observed in \cite[Sec. 3]{linearizedRS} and \cite[Subsec. IV-A]{secure-multishot}, linearized Reed-Solomon codes recover Gabidulin codes \cite{gabidulin, roth} when $ g=1 $, and they recover Reed-Solomon codes \cite{reed-solomon} when $ m = r_1 = r_2 = \ldots = r_g = 1 $. These are the cases when the sum-rank metric particularizes to the rank metric and Hamming metric, respectively. The second choice of parameters explains why setting $ m = r_1 = r_2 = \ldots = r_g = 1 $ in this paper recovers Reed-Solomon codes with local replication (one-dimensional local codes).

\section{Main Construction of MR-LRCs} \label{sec main construction}

In this section, we briefly recall the definitions of locally repairable codes \cite{gopalan, kamath, kadhe, zeh-multiple, chen-hao, kim} and maximal recoverability \cite{blaum-RAID, gopalan-MR}, and give our main construction. Proofs and further properties of our construction are left to the following sections.

Let $ \mathbb{F} $ be a finite field. In this work, we will consider disjoint local groups, which is usual in the maximal recoverability or PMDS literature \cite{blaum-RAID, gopalan-MR}.

\begin{definition}[\textbf{Locally repairable codes}] \label{def LRC}
Fix integers $ g, r_i, \delta_i \geq 1 $, for $ i = 1,2, \ldots, g $. We say that a code $ \mathcal{C} \subseteq \mathbb{F}^n $ is an $ (n,k) $ locally repairable code (LRC) with $ (\Gamma_i, r_i, \delta_i)_{i=1}^g $-localities, or $ (r_i, \delta_i)_{i=1}^g $-localities for short, if $ k = \log_{|\mathbb{F}|}|\mathcal{C}| $, $ [n] = \Gamma_1 \cup \Gamma_2 \cup \ldots \cup \Gamma_g $, $ \Gamma_i \cap \Gamma_j = \varnothing $ if $ i \neq j $ (that is, the sets $ \Gamma_1, \Gamma_2, \ldots, \Gamma_g $ form a partition of $ [n] $), and 
\begin{enumerate}
\item
$ | \Gamma_i | \leq r_i + \delta_i - 1 $,
\item
$ {\rm d}_H(\mathcal{C}_{\Gamma_i}) \geq \delta_i $,
\end{enumerate}
for $ i = 1,2, \ldots, g $. The set $ \Gamma_i $ is called the $ i $th local group. In many occasions, we only use the term \textit{locality} for the number $ r_i $, whereas $ \delta_i $ is called the \textit{local distance}.
\end{definition}

Figs. \ref{fig LRCs equal} and \ref{fig LRCs unequal} below depict systematic LRCs.

\begin{figure} [!h]
\hspace*{-1.8em}
\begin{center}
\begin{tabular}{c@{\extracolsep{1cm}}c}
\begin{tikzpicture}[
square/.style = {draw, rectangle, 
                 minimum size=\m, outer sep=0, inner sep=0, font=\small,
                 },
                        ]
\def\m{18pt}
\def\w{7}
\def\h{5}
\def\loc{3}
    \pgfmathsetmacro\uw{int(\w/2)}
    \pgfmathsetmacro\uh{int(\h/2)}

\def\i{10}
  \foreach \x in {1,...,\w}
    \foreach \y in {1,...,\h}
       {    
           \ifnum\y>\loc
               \node [square, fill=gray!50]  (\x,\y) at (\x*\m + \i*\m,-\y*\m) {$ c^{(\x)}_{\y} $};
           \else
               \node [square, fill=white]  (\x,\y) at (\x*\m + \i*\m,-\y*\m) {$ x^{(\x)}_{\y} $};
           \fi
       }
       
  \foreach \y in {1,...,\loc}
       {
           \node [square, fill=black!70]  (\w,\y) at (\w*\m + \i*\m,-\y*\m) {{\color{white}$ c^{(\w)}_{\y} $}};
       }

   \draw [decorate, thick,decoration={brace,amplitude=5pt}]
   (10.2*\m,-3.5*\m) -- (10.2*\m,-0.5*\m) node[midway,xshift=-3.3em]{\begin{tabular}{c}Information\\ symbols\\ (white) \end{tabular}}; 

   \draw [decorate, thick,decoration={brace,amplitude=5pt}]
   (10.2*\m,-5.5*\m) -- (10.2*\m,-3.5*\m) node[midway,xshift=-3.3em]{\begin{tabular}{c}Local\\ parities \\ (light grey) \end{tabular}};    

   \draw [decorate, thick,decoration={brace,amplitude=5pt}]
   (17.8*\m,-0.5*\m) -- (17.8*\m,-3.5*\m) node[midway,xshift=3.3em]{\begin{tabular}{c}Global\\ parities\\ (dark grey) \end{tabular}}; 

   \node at (16.5*\m,-6*\m) {\rotatebox[origin=c]{90}{$\Rsh$} Each column forms a local codeword};

\end{tikzpicture}

\end{tabular}
\end{center}

\caption{Illustration of a systematic LRC with equal localities ($ r = 3 $) and local distances ($ \delta = 3 $), which allows to represent symbols in a rectangular array. Each box represents a node storing a symbol in $ \mathbb{F} $. The $ i $th column forms the symbols stored in the $ i $th local group $ \Gamma_i $.  The $ x $'s denote information symbols and the $ c $'s denote parities. Each local group has $ 2 $ local parities, hence any $ \delta - 1 = 2 $ erasures inside a column can be corrected from the remaining $ r = 3 $ symbols in that column. The dimension is $ k = 18 $. The number of global parities is $ h = gr - k = 3 $. These parities allow us to correct erasure patterns where more than two erasures occur in a single column (see Fig. \ref{fig erasure pattern for MR}).   }
\label{fig LRCs equal}

\end{figure}
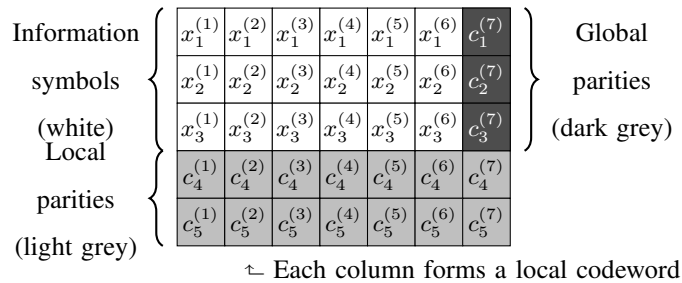

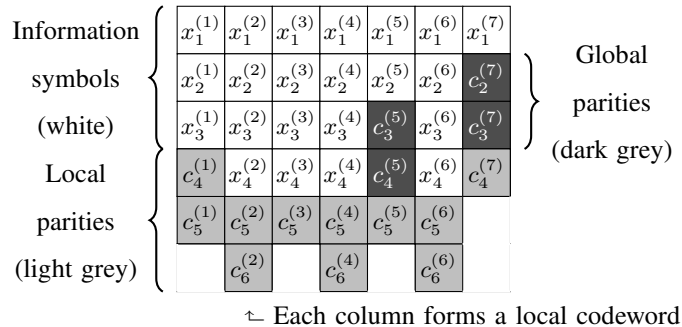
\begin{figure} [!h]
\hspace*{-1.8em}
\begin{center}
\begin{tabular}{c@{\extracolsep{1cm}}c}
\begin{tikzpicture}[
square/.style = {draw, rectangle, 
                 minimum size=\m, outer sep=0, inner sep=0, font=\small,
                 },
                        ]
\def\m{18pt}
\def\w{7}
\def\h{6}
\def\loc{4}
    \pgfmathsetmacro\uw{int(\w/2)}
    \pgfmathsetmacro\uh{int(\h/2)}

\def\i{10}
  \foreach \x in {1,...,\w}
    \foreach \y in {1,...,\h}
       {    
           \ifnum\y>\loc
               \node [square, fill=gray!50]  (\x,\y) at (\x*\m + \i*\m,-\y*\m) {$ c^{(\x)}_{\y} $};
           \else
               \node [square, fill=white]  (\x,\y) at (\x*\m + \i*\m,-\y*\m) {$ x^{(\x)}_{\y} $};
           \fi
       }
       
  \foreach \y in {2,...,\loc}
       {
           \node [square, fill=black!70]  (\w,\y) at (\w*\m + \i*\m,-\y*\m) {{\color{white}$ c^{(\w)}_{\y} $}};
       }
       
  \foreach \y in {3,...,\loc}
       {
           \node [square, fill=black!70]  (5,\y) at (5*\m + \i*\m,-\y*\m) {{\color{white}$ c^{(5)}_{\y} $}};
       }

   \draw [decorate, thick,decoration={brace,amplitude=5pt}]
   (10.2*\m,-3.5*\m) -- (10.2*\m,-0.5*\m) node[midway,xshift=-3.3em]{\begin{tabular}{c}Information\\ symbols\\ (white) \end{tabular}}; 

   \draw [decorate, thick,decoration={brace,amplitude=5pt}]
   (10.2*\m,-6.5*\m) -- (10.2*\m,-3.5*\m) node[midway,xshift=-3.3em]{\begin{tabular}{c}Local\\ parities \\ (light grey) \end{tabular}};    

   \draw [decorate, thick,decoration={brace,amplitude=5pt}]
   (17.8*\m,-1.5*\m) -- (17.8*\m,-3.5*\m) node[midway,xshift=3.3em]{\begin{tabular}{c}Global\\ parities\\ (dark grey) \end{tabular}}; 

   \node at (16.5*\m,-7*\m) {\rotatebox[origin=c]{90}{$\Rsh$} Each column forms a local codeword};

\node [square, fill=gray!50]  at (1*\m + \i*\m,-4*\m) {$ c^{(1)}_4 $};

\node [square, fill=white!100]  at (1*\m + \i*\m,-6*\m) {$  $};
\pic [square, fill=white!100] at (1*\m + \i*\m,-6*\m) {lower={white}{$ $}};
\pic [square, fill=white!100] at (1*\m + \i*\m,-6*\m) {left={white}{$ $}};

\node [square, fill=white!100]  at (3*\m + \i*\m,-6*\m) {$  $};
\pic [square, fill=white!100] at (3*\m + \i*\m,-6*\m) {lower={white}{$ $}};

\node [square, fill=white!100]  at (5*\m + \i*\m,-6*\m) {$  $};
\pic [square, fill=white!100] at (5*\m + \i*\m,-6*\m) {lower={white}{$ $}};

\node [square, fill=gray!50]  at (7*\m + \i*\m,-4*\m) {$ c^{(7)}_4 $};

\node [square, fill=white!100]  at (7*\m + \i*\m,-5*\m) {$  $};
\pic [square, fill=white!100] at (7*\m + \i*\m,-5*\m) {lower={white}{$ $}};
\pic [square, fill=white!100] at (7*\m + \i*\m,-5*\m) {right={white}{$ $}};

\node [square, fill=white!100]  at (7*\m + \i*\m,-6*\m) {$  $};
\pic [square, fill=white!100] at (7*\m + \i*\m,-6*\m) {upper={white}{$ $}};
\pic [square, fill=white!100] at (7*\m + \i*\m,-6*\m) {lower={white}{$ $}};
\pic [square, fill=white!100] at (7*\m + \i*\m,-6*\m) {right={white}{$ $}};

\end{tikzpicture}

\end{tabular}
\end{center}

\caption{Illustration of a systematic LRC with unequal localities and local distances, with the same notation as in Fig. \ref{fig LRCs equal}. For instance, the first local group has locality $ r_1 = 3 $ and local distance $ \delta_1 = 3 $, whereas for the third local group, $ r_3 = 4 $ and $ \delta_3 = 2 $. Note that global parities may be arbitrarily distributed among the local groups.  }
\label{fig LRCs unequal}

\end{figure}

Every code is an LRC for any partition of $ [n] $ if $ \delta_i = 1 $ for $ i = 1,2, \ldots, g $, which includes locality but not repair. Every code with distance $ d $ is also an LRC for $ \Gamma_1 = [n] $, $ r_1 = n - d + 1 $ and $ \delta_1 = d $, which includes repair but not locality. 

An $ (n,k) $ MDS code can only have these types of locality, and localities where a local group with distance $ \delta_i > 1 $ must satisfy $ r_i \geq k $. To see this, just observe that any other type of localities imply that there exists some set of $ k $ symbols with some redundancy, thus cannot be an information set. For this reason, MDS codes are not good candidates as LRCs.

Finally, observe that $ r_1 = r_2 = \ldots = r_g = 1 $ means $ \delta_i $-replication of the $ i $th symbol.

We now extend the concept of maximal recoverability from \cite[Def. 2.1]{blaum-RAID} and \cite[Def. 6]{gopalan-MR} to unequal localities and local distances.

\begin{definition} [\textbf{Maximal recoverability}] \label{def MR}
We say that an LRC $ \mathcal{C} \subseteq \mathbb{F}^n $ with $ (\Gamma_i, r_i, \delta_i)_{i=1}^g $-localities is maximally recoverable (MR) if, for any $ \Delta_i \subseteq \Gamma_i $ with $ | \Gamma_i \setminus \Delta_i | = \delta_i - 1 $, for $ i = 1,2, \ldots,g $, the restricted code $ \mathcal{C}_\Delta \subseteq \mathbb{F}^{|\Delta|} $ is MDS, where $ \Delta = \bigcup_{i=1}^g \Delta_i $.
\end{definition} 

An example of an erasure pattern correctable by an MR-LRC is depicted in Fig. \ref{fig erasure pattern for MR} below.

\begin{figure} [!h]
\hspace*{-1.8em}
\begin{center}
\begin{tabular}{c@{\extracolsep{1cm}}c}
\begin{tikzpicture}[
square/.style = {draw, rectangle, 
                 minimum size=\m, outer sep=0, inner sep=0, font=\small,
                 },
                        ]
\def\m{18pt}
\def\w{7}
\def\h{6}
\def\loc{4}
    \pgfmathsetmacro\uw{int(\w/2)}
    \pgfmathsetmacro\uh{int(\h/2)}

\def\i{10}
  \foreach \x in {1,...,\w}
    \foreach \y in {1,...,\h}
       {    
           \ifnum\y>\loc
               \node [square, fill=gray!50]  (\x,\y) at (\x*\m + \i*\m,-\y*\m) {$ c^{(\x)}_{\y} $};
           \else
               \node [square, fill=white]  (\x,\y) at (\x*\m + \i*\m,-\y*\m) {$ x^{(\x)}_{\y} $};
           \fi
       }
       
  \foreach \y in {2,...,\loc}
       {
           \node [square, fill=black!70]  (\w,\y) at (\w*\m + \i*\m,-\y*\m) {{\color{white}$ c^{(\w)}_{\y} $}};
       }
       
  \foreach \y in {3,...,\loc}
       {
           \node [square, fill=black!70]  (5,\y) at (5*\m + \i*\m,-\y*\m) {{\color{white}$ c^{(5)}_{\y} $}};
       }

   \draw [decorate, thick,decoration={brace,amplitude=5pt}]
   (10.2*\m,-3.5*\m) -- (10.2*\m,-0.5*\m) node[midway,xshift=-3.3em]{\begin{tabular}{c}Information\\ symbols\\ (white) \end{tabular}}; 

   \draw [decorate, thick,decoration={brace,amplitude=5pt}]
   (10.2*\m,-6.5*\m) -- (10.2*\m,-3.5*\m) node[midway,xshift=-3.3em]{\begin{tabular}{c}Local\\ parities \\ (light grey) \end{tabular}};    

   \draw [decorate, thick,decoration={brace,amplitude=5pt}]
   (17.8*\m,-1.5*\m) -- (17.8*\m,-3.5*\m) node[midway,xshift=3.3em]{\begin{tabular}{c}Global\\ parities\\ (dark grey) \end{tabular}};

\node [square, fill=gray!50]  at (1*\m + \i*\m,-4*\m) {$ c^{(1)}_4 $};

\node [square, fill=white!100]  at (1*\m + \i*\m,-6*\m) {$  $};
\pic [square, fill=white!100] at (1*\m + \i*\m,-6*\m) {lower={white}{$ $}};
\pic [square, fill=white!100] at (1*\m + \i*\m,-6*\m) {left={white}{$ $}};

\node [square, fill=white!100]  at (3*\m + \i*\m,-6*\m) {$  $};
\pic [square, fill=white!100] at (3*\m + \i*\m,-6*\m) {lower={white}{$ $}};

\node [square, fill=white!100]  at (5*\m + \i*\m,-6*\m) {$  $};
\pic [square, fill=white!100] at (5*\m + \i*\m,-6*\m) {lower={white}{$ $}};

\node [square, fill=gray!50]  at (7*\m + \i*\m,-4*\m) {$ c^{(7)}_4 $};

\node [square, fill=white!100]  at (7*\m + \i*\m,-5*\m) {$  $};
\pic [square, fill=white!100] at (7*\m + \i*\m,-5*\m) {lower={white}{$ $}};
\pic [square, fill=white!100] at (7*\m + \i*\m,-5*\m) {right={white}{$ $}};

\node [square, fill=white!100]  at (7*\m + \i*\m,-6*\m) {$  $};
\pic [square, fill=white!100] at (7*\m + \i*\m,-6*\m) {upper={white}{$ $}};
\pic [square, fill=white!100] at (7*\m + \i*\m,-6*\m) {lower={white}{$ $}};
\pic [square, fill=white!100] at (7*\m + \i*\m,-6*\m) {right={white}{$ $}};

\node [square, fill=white]  at (1*\m + \i*\m,-3*\m) {$ \spadesuit $}; 
\node [square, fill=gray!50]  at (1*\m + \i*\m,-5*\m) {$ \spadesuit $};  
 
\node [square, fill=white]  at (2*\m + \i*\m,-2*\m) {$ \spadesuit $}; 
\node [square, fill=gray!50]  at (2*\m + \i*\m,-5*\m) {$ \spadesuit $}; 

\node [square, fill=white]  at (3*\m + \i*\m,-4*\m) {$ \spadesuit $}; 

\node [square, fill=white]  at (4*\m + \i*\m,-1*\m) {$ \spadesuit $}; 
\node [square, fill=white]  at (4*\m + \i*\m,-4*\m) {$ \spadesuit $}; 

\node [square, fill=white]  at (5*\m + \i*\m,-2*\m) {$ \spadesuit $}; 

\node [square, fill=gray!50]  at (6*\m + \i*\m,-5*\m) {$ \spadesuit $}; 
\node [square, fill=gray!50]  at (6*\m + \i*\m,-6*\m) {$ \spadesuit $}; 

\node [square, fill=black!70]  at (7*\m + \i*\m,-3*\m) {{\color{white}$ \spadesuit $}}; 

\node [square, fill=white]  at (2*\m + \i*\m,-3*\m) {$ \clubsuit $}; 
\node [square, fill=white]  at (3*\m + \i*\m,-1*\m) {$ \clubsuit $}; 
\node [square, fill=gray!50]  at (4*\m + \i*\m,-5*\m) {$ \clubsuit $}; 
\node [square, fill=white]  at (7*\m + \i*\m,-1*\m) {$ \clubsuit $};

\end{tikzpicture}

\end{tabular}
\end{center}

\caption{Illustration of an erasure pattern correctable by an MR-LRC with parameters as in Fig. \ref{fig LRCs unequal}. The erasure patterns are those consisting of $ \delta_i-1 $ erasures in the $ i $th local group (depicted by $ \spadesuit $), plus any $ h = N -k = 4 $ extra erasures placed anywhere (depicted by $ \clubsuit $), where $ N = r_1 + r_2 + \cdots + r_g $. This is because, after removing the $ \delta_i-1 $ erasures from the $ i $th local group, the restricted code is an $ (N,k) $ MDS code, and hence must be able to correct any $ h $ erasures in the remaining $ N $ nodes.  }
\label{fig erasure pattern for MR}

\end{figure}
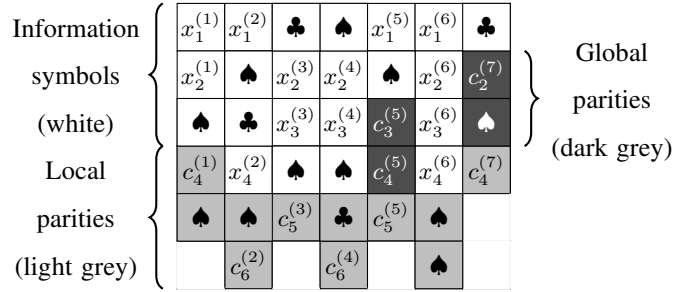

We next introduce our construction of MR-LRCs based on linearized Reed-Solomon codes (Definition \ref{def lin RS codes}). 

\begin{construction} \label{construction 1}
Fix the initial number of local groups $ g $ and initial localities $ r_1 $, $ r_2 $, \ldots, $ r_g $. Now choose any \textit{base field} size $ q $ and any extension degree $ m $ satisfying $ q > g $ and $ m \geq \max_i r_i $, and define the \textit{global field} $ \mathbb{F} = \mathbb{F}_{q^m} $. Next choose:
\begin{enumerate}
\item
\textit{Outer code}: Any $ (N,k) $ code $ \mathcal{C}_{out} \subseteq \mathbb{F}_{q^m}^N $ that is MSRD over $ \mathbb{F}_q $ for the sum-rank length $ N = \sum_{i=1}^g r_i $, such as a linearized Reed-Solomon code (Definition \ref{def lin RS codes}).
\item
\textit{Local codes}: Any $ (r_i+\delta_i-1, r_i) $ MDS code $ \mathcal{C}_{loc}^{(i)} \subseteq \mathbb{F}_{q_i}^{r_i + \delta_i - 1} $, linear over \textit{local fields} $ \mathbb{F}_{q_i} $, where $ q $ is a power of $ q_i $, for $ i = 1,2, \ldots, g $.
\end{enumerate}
The \textit{global code} is then given as follows.
\begin{enumerate}
\item[3)]
\textit{Global code}: Let $ \mathcal{C}_{glob} \subseteq \mathbb{F}_{q^m}^n $, with $ n = \sum_{i=1}^g (r_i + \delta_i - 1) = N + \sum_{i=1}^g (\delta_i - 1) $, be given by
$$ \mathcal{C}_{glob} = \mathcal{C}_{out} \diag(A_1, A_2, \ldots, A_g), $$
where $ A_i \in \mathbb{F}_{q_i}^{r_i \times (r_i + \delta_i - 1)} $ is any generator matrix of $ \mathcal{C}_{loc}^{(i)} $, for $ i = 1,2, \ldots, g $.
\end{enumerate}
The encoding procedure for $ \mathcal{C}_{glob} $ using first the outer code and then the local codes is depicted in Fig. \ref{fig encoding for construction I}. 
\end{construction}

\begin{figure} [!h]
\begin{center}
\begin{tabular}{c@{\extracolsep{1cm}}c}
\begin{tikzpicture}[
square/.style = {draw, rectangle, 
                 minimum size=\m, outer sep=8, inner sep=8, font=\normalsize,
                 },
                        ]
\def\m{25pt}

\node at (0,1.5) {$ \mathbf{x} \in \mathbb{F}_{q^m}^k $};

\draw[|-, ultra thick] (0,1) -- (0,0.5);

\node [square, thick] at (0,0) {Outer code $ \mathcal{C}_{out} \subseteq \mathbb{F}_{q^m}^N $};

\draw[-stealth, ultra thick] (0,-0.5) -- (0,-1);

\node [] at (0,-1.5) {$ \mathbf{c}_{out} = (\mathbf{c}^{(1)}, \mathbf{c}^{(2)}, \ldots, \mathbf{c}^{(g)}) \in \mathbb{F}_{q^m}^N $};

\draw[|-, ultra thick] (-1.15,-2) -- (-1.25,-2.5);

\draw[-stealth, ultra thick] (-1.45,-3.5) -- (-1.55,-4);

\node [square, thick, fill=white] at (-2,-3) {$ \mathcal{C}_{loc}^{(1)} \subseteq \mathbb{F}_{q_i}^{r_1 + \delta_1 - 1} $};

\node [] at (0,-3) {$ \ldots $};

\draw[|-, ultra thick] (1.15,-2) -- (1.25,-2.5);

\draw[-stealth, ultra thick] (1.45,-3.5) -- (1.55,-4);

\node [square, thick, fill = white] at (2,-3) {$ \mathcal{C}_{loc}^{(g)} \subseteq \mathbb{F}_{q_i}^{r_g + \delta_g - 1} $};

\node [] at (0,-4.8) {$ \mathbf{c}_{glob} = ( \underbrace{\mathbf{c}^{(1)} A_1}_\textrm{Local group 1} , \underbrace{\mathbf{c}^{(2)} A_2}_\textrm{Local group 2}, \ldots, \underbrace{\mathbf{c}^{(g)} A_g}_\textrm{Local group g}) \in \mathbb{F}_{q^m}^n $};

\end{tikzpicture}

\end{tabular}
\end{center}

\caption{Illustration of the encoding procedure for $ \mathcal{C}_{glob} $ in Construction \ref{construction 1}. Let $ \mathbf{x} \in \mathbb{F}_{q^m}^k $ be $ k $ symbols over $ \mathbb{F}_{q^m} $ of the file. We first encode them with the outer code $ \mathcal{C}_{out} $ to form $ \mathbf{c}_{out} \in \mathbb{F}_{q^m}^N $. We then partition the outer codeword as $ \mathbf{c}_{out} = (\mathbf{c}^{(1)}, \mathbf{c}^{(2)}, \ldots, \mathbf{c}^{(g)}) $, where $ \mathbf{c}^{(i)} \in \mathbb{F}_{q^m}^{r_i} $. Finally, we encode each $ \mathbf{c}^{(i)} $ using a generator matrix $ A_i \in \mathbb{F}_{q_i}^{r_i \times (r_i + \delta_i - 1)} $ of the $ i $th local code, and store $ \mathbf{c}^{(i)} A_i \in \mathbb{F}_{q^m}^{r_i + \delta_i - 1} $ in the $ i $th local group of nodes.  }
\label{fig encoding for construction I}

\end{figure}
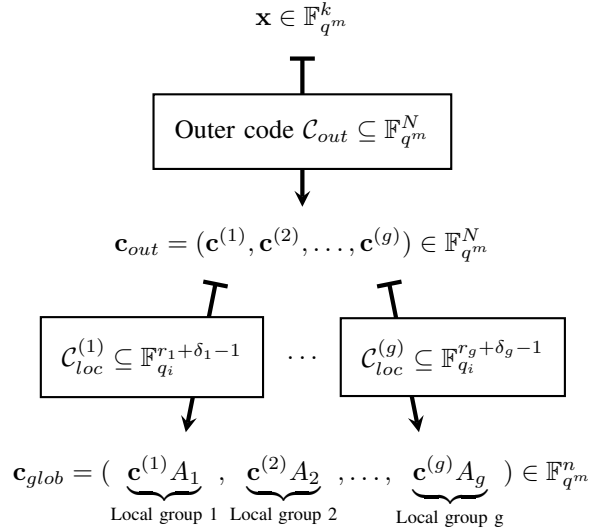

\begin{remark}
Typically, the local field sizes must satisfy $ q_i \geq r_i + \delta_i - 1 $ so that we may choose MDS local codes (for instance, Reed-Solomon codes \cite{reed-solomon}). However, if $ \delta_i = 2 $, we may always choose $ q_i = 2 $ if $ 2 \mid q $, and local repair in the $ i $th local group can be performed by XORing. 
\end{remark}

Observe that the difference with \cite{rawat, kadhe, kim, calis} is that Gabidulin codes do not exist for the parameters described in Construction \ref{construction 1} (they require $ m \geq N = \sum_{i=1}^g r_i $), whereas we may use linearized Reed-Solomon codes (Definition \ref{def lin RS codes}) for such parameters, which are still MSRD. See Section \ref{sec comparisons} for detailed comparisons of global field sizes.

The following main result follows from Corollary \ref{cor MSRD equivalent to MR LRC}, which will be proven in Subsection \ref{subsec architecture and constr}. The minimum distance of the global code in Construction \ref{construction 1} will be estimated in Theorem \ref{th global distance and threshold}, Subsection \ref{subsec global distance arbitrary local codes}.

\begin{theorem} \label{th constr 1 is MR}
Let $ \mathcal{C}_{glob} \subseteq \mathbb{F}_{q^m}^n $ be the global code from Construction \ref{construction 1}, and let $ \Gamma_i \subseteq [n] $ be the subset of coordinates ranging from $ \sum_{j=1}^{i-1} (r_j+\delta_j-1) + 1 $ to $ \sum_{j=1}^i (r_j + \delta_j - 1) $, for $ i = 1,2, \ldots, g $. Then the code $ \mathcal{C}_{glob} \subseteq \mathbb{F}_{q^m}^n $ has $ (\Gamma_i, r_i, \delta_i)_{i=1}^g $-localities and is maximally recoverable.\usetikzlibrary{shapes.misc, fit}
\end{theorem}

We conclude by noting that we may easily find a systematic form of the global code in Construction \ref{construction 1}. Here \textit{systematic} only means that certain $ k $ symbols of each codeword form the uncoded file, but they are not necessarily the first $ k $ symbols. Actually, we may distribute the $ k $ uncoded symbols as wanted among the local groups, up to the locality of each group (as shown, for instance, in Fig. \ref{fig LRCs unequal}). This more general systematic form is of interest if certain local groups are required to store a certain part of the original file.

Since $ \mathcal{C}_{out} \subseteq \mathbb{F}_{q^m}^N $ is MSRD, then it is MDS, thus any set of $ k $ coordinates is an information set. Since $ N = \sum_{i=1}^g r_i $, we may partition $ k = \sum_{i=1}^g k_i $, such that $ 0 \leq k_i \leq r_i $, for $ i = 1,2, \ldots, g $. We may then find a systematic generator matrix of $ \mathcal{C}_{out} $ of the form
\begin{equation}
G = (\widetilde{I}_{k_1}, G_1 | \widetilde{I}_{k_2}, G_2 | \ldots | \widetilde{I}_{k_g}, G_g) \in \mathbb{F}_{q^m}^{k \times N}, 
\label{eq systematic generator linRS}
\end{equation}
where $ G_i \in \mathbb{F}_{q^m}^{k \times (r_i - k_i)} $, and $ \widetilde{I}_{k_i} \in \mathbb{F}_{q^m}^{k \times k_i} $ is identically zero except for the rows in the $ i $th block of $ k_i $ rows, where it is the $ k_i \times k_i $ identity matrix, for $ i = 1,2, \ldots, g $.

If $ A = \diag(A_1, A_2, \ldots, A_g) \in \mathbb{F}_q^{N \times n} $ is such that $ A_i = (I_{r_i}, B_i) \in \mathbb{F}_q^{r_i \times n_i} $ is systematic, for $ i = 1,2, \ldots, g $, then
\begin{equation} 
GA = (\widetilde{I}_{k_1}, \widetilde{G}_1 | \widetilde{I}_{k_2}, \widetilde{G}_2 | \ldots | \widetilde{I}_{k_g}, \widetilde{G}_g) \in \mathbb{F}_{q^m}^{k \times n}
\label{eq global systematic generator}
\end{equation}
is a systematic generator matrix of $ \mathcal{C}_{glob} \subseteq \mathbb{F}_{q^m}^n $, where $ \widetilde{G}_i = (G_i, (\widetilde{I}_{k_i}, G_i) B_i) \in \mathbb{F}_{q^m}^{k \times (n_i - k_i)} $, for $ i = 1,2, \ldots, g $. 

Finally, note that this systematic encoding procedure follows the same steps as in Fig. \ref{fig encoding for construction I}. We first add the global parities and then the local parities.

\section{MR-LRCs with any Local Linear Codes} \label{sec universal lrc arbitrary local codes}

In this section, we study LRCs where local groups are disjoint, but locally encoded with arbitrary linear codes over some subfield $ \mathbb{F}_q \subseteq \mathbb{F} $. We will give the connection between MSRD codes and MR-LRCs in Subsection \ref{subsec architecture and constr}, and we will study global distances in Subsection \ref{subsec global distance arbitrary local codes}. As a consequence, we show that Construction \ref{construction 1} gives MR-LRCs (thus LRCs with optimal global distance) for any choice of local linear codes.  

As shown later in Section \ref{sec dynamic properties}, a direct application of this study, among others, will be partitioning the local MDS codes into Cartesian products of shorter local MDS codes, over smaller fields, to modify localities dynamically and adapt the DSS to new hot and cold data, or to obtain hierarchical  MR-LRCs.

\subsection{General MR-LRCs and MSRD Codes} \label{subsec architecture and constr}

In this subsection, we show that LRCs with disjoint local linear codes always have the architecture of Construction \ref{construction 1} (Lemma \ref{lemma equivalence linear local codes}), depicted in Fig. \ref{fig encoding for construction I}. We then show that MSRD outer codes achieve MR simultaneously for all families of local linear codes (Corollary \ref{cor MSRD equivalent to MR LRC}). We conclude by showing that the maximum locality $ r = \max_i r_i $ is the smallest extension degree of $ \mathbb{F} $ over $ \mathbb{F}_q $ satisfying this property, which is achieved by Construction \ref{construction 1} (Proposition \ref{prop minimum m for universal LRC}).

Fix a subfield $ \mathbb{F}_q \subseteq \mathbb{F} $. The proof of the following lemma is straightforward by linear algebra, and is left to the reader.

\begin{lemma} \label{lemma equivalence linear local codes}
Let $ \mathcal{C}_{glob} \subseteq \mathbb{F}^n $ be a (linear or non-linear) code, where $ n = n_1 + n_2 + \cdots + n_g $, and let $ \Gamma_i $ be the set of coordinates ranging from $ \sum_{j=1}^{i-1} n_j + 1 $ to $ \sum_{j=1}^i n_j $, for $ i = 1,2, \ldots, g $. The following are equivalent:
\begin{enumerate}
\item
$ (\mathcal{C}_{glob})_{\Gamma_i} \subseteq \mathcal{C}_i $, where $ \mathcal{C}_i \subseteq \mathbb{F}^{n_i} $ is an $ r_i $-dimensional linear code with a generator matrix with coefficients in $ \mathbb{F}_q $, for $ i = 1,2, \ldots, g $.
\item
There exist a full-rank matrix $ H_i \in \mathbb{F}_q^{(n_i - r_i) \times n_i} $ such that $ \mathbf{c}_{\Gamma_i} H_i = \mathbf{0} $, for all $ \mathbf{c} \in \mathcal{C}_{glob} $ and all $ i = 1,2, \ldots, g $.
\item 
There exist full-rank matrices $ A_i \in \mathbb{F}_q^{r_i \times n_i} $ with $ 1 \leq r_i \leq n_i $, for $ i = 1,2, \ldots, g $, such that
\begin{equation}
\mathcal{C}_{glob} = \mathcal{C}_{out} \diag(A_1, A_2, \ldots, A_g),
\label{eq architecture arbitrary local codes}
\end{equation}
for some outer code $ \mathcal{C}_{out} \subseteq \mathbb{F}^N $, where $ N = \sum_{i=1}^g r_i $ and $ | \mathcal{C}_{out} | = | \mathcal{C}_{glob} | $. Moreover, $ \mathcal{C}_{glob} $ is linear if, and only if, $ \mathcal{C}_{out} $ is linear.
\end{enumerate}
The relation between these items is that $ A_i $ and $ H_i $ are generator and parity-check matrices of $ \mathcal{C}_i $, respectively, for $ i = 1,2, \ldots, g $.
\end{lemma}

Encoding with the code $ \mathcal{C}_{glob} $ satisfying the conditions in Lemma \ref{lemma equivalence linear local codes} also follows the steps in Fig. \ref{fig encoding for construction I}. We only need to replace $ r_i + \delta_i - 1 $ by $ n_i $ and choose $ \mathcal{C}_{loc}^{(i)} $ as the subfield subcode of $ \mathcal{C}_i $ over $ \mathbb{F}_{q_i} $, for $ i = 1,2, \ldots, g $, which need not be MDS. 

By Item 2, codes with this structure are included among those described in \cite[Def. 2.1]{grid-like}. By Item 1, they are included among those in Definition \ref{def LRC}, and by Item 3, they include Construction \ref{construction 1}. LRCs with non-MDS local linear codes have also been considered recently in \cite{huang-MELRC, blaum-MELRC}, where they are called \textit{Multi-Erasure LRCs}. Their approach is however focused on product codes, and the MR condition is not pursued.

We deduce the following two consequences.

\begin{corollary}
With notation as in Lemma \ref{lemma equivalence linear local codes}, any erasure pattern $ \mathcal{E}_i \subseteq \Gamma_i $ that can be corrected by the local code $ \mathcal{C}_i $, can be corrected by the global code $ \mathcal{C}_{glob} $ with the same complexity over the same field as with $ \mathcal{C}_i $, for $ i = 1,2, \ldots, g $.
\end{corollary}

\begin{corollary} \label{cor max k for initial localities}
With notation as in Lemma \ref{lemma equivalence linear local codes}, it holds that 
\begin{equation}
k = \log_{|\mathbb{F}|}| \mathcal{C}_{glob} | \leq N = \sum_{i=1}^g r_i.
\label{eq upper bound on k info rate}
\end{equation}
\end{corollary}

The previous corollary motivates the following definition.

\begin{definition} \label{def global and local parities}
With notation as in Lemma \ref{lemma equivalence linear local codes}, we call $ h = \sum_{i=1}^g r_i - k \geq 0 $ the \textit{number of global parities} of the global code $ \mathcal{C}_{glob} $, which coincides with the number of conventional parities of the outer code $ \mathcal{C}_{out} $. The other $ \sum_{i=1}^g (n_i - r_i) $ parities of $ \mathcal{C}_{glob} $ are given by the parities of the local codes, and therefore are called \textit{local parities}. See also Figs. \ref{fig LRCs equal} and \ref{fig LRCs unequal} for a graphical description.
\end{definition}

The following definition is a natural extension of Definition \ref{def LRC} for arbitrary disjoint local linear codes. 

\begin{definition} \label{def locally linear LRC}
Let $ \mathcal{C}_{glob} \subseteq \mathbb{F}^n $ be a (linear or non-linear) code, where $ n = n_1 + n_2 + \cdots + n_g $ and define $ \Gamma_i $ as the set of coordinates ranging from $ \sum_{j=1}^{i-1} n_j + 1 $ to $ \sum_{j=1}^i n_j $, for $ i = 1,2, \ldots, g $. We say that $ \mathcal{C}_{glob} $ is a $ (\Gamma_i, \mathcal{C}_i)_{i=1}^g $-LRC if the equivalent conditions in Lemma \ref{lemma equivalence linear local codes} hold.
\end{definition}

We now characterize the global erasure patterns that are information-theoretically correctable. This holds in particular when the local codes are MDS. 

\begin{theorem} \label{th optimal arbitrary local codes}
Fix an $ (n,k) $ and $ (\Gamma_i, \mathcal{C}_i)_{i=1}^g $-LRC $ \mathcal{C}_{glob} \subseteq \mathbb{F}^n $ as in Definition \ref{def locally linear LRC}. Let $ \mathcal{E} \subseteq [n] $ be an erasure pattern, and define $ \mathcal{E}_i = \mathcal{E} \cap \Gamma_i $ and $ \mathcal{R}_i = \Gamma_i \setminus \mathcal{E}_i $, for $ i = 1,2, \ldots, g $. The following hold:
\begin{enumerate}
\item
If $ \sum_{i=1}^g {\rm Rk}(A_i|_{\mathcal{R}_i}) < k $, then the erasure pattern cannot be corrected by $ \mathcal{C}_{glob} $ for all codewords $ \mathbf{c} \in \mathcal{C}_{glob} $, independently of what outer code $ \mathcal{C}_{out} $ is used.
\item
If $ \sum_{i=1}^g {\rm Rk}(A_i|_{\mathcal{R}_i}) \geq k $ and $ \mathcal{C}_{out} $ is an MSRD code over $ \mathbb{F}_q $ for the sum-rank length partition $ N = \sum_{i=1}^g r_i $, then the erasure pattern can be corrected by $ \mathcal{C}_{glob} $ for all codewords $ \mathbf{c} \in \mathcal{C}_{glob} $.
\end{enumerate}
\end{theorem}
\begin{proof}
We prove each item separately:

1) Assume that there exists a decoder $ D : \mathcal{C}_{glob}|_\mathcal{R} \longrightarrow \mathcal{C}_{glob} $, where $ \mathcal{R} = [n] \setminus \mathcal{E} $, such that $ D(\mathbf{c}_\mathcal{R}) = \mathbf{c} $, for all $ \mathbf{c} \in \mathcal{C}_{glob} $. Let $ A |_\mathcal{R} = \diag(A_1 |_{\mathcal{R}_1}, A_2 |_{\mathcal{R}_2}, \ldots, A_g |_{\mathcal{R}_g}) $. Then the decoder can be rewritten as
$$ D : \mathcal{C}_{out} (A |_\mathcal{R}) \longrightarrow \mathcal{C}_{out}, $$
where $ D(\mathbf{c} (A |_\mathcal{R})) = \mathbf{c} $, for all $ \mathbf{c} \in \mathcal{C}_{out} $. In particular, $ D : \mathcal{C}_{out} (A |_\mathcal{R}) \longrightarrow \mathcal{C}_{out} $ is a bijective map.

Fix $ i = 1,2, \ldots, g $ and let $ s_i = \Rk (A_i|_{\mathcal{R}_i}) $. There exists $ \mathcal{S}_i \subseteq \mathcal{R}_i $ such that $ |\mathcal{S}_i| = \Rk (A_i|_{\mathcal{S}_i}) = s_i $. Hence the restriction map $ \pi : \mathcal{C}_{out} (A |_\mathcal{R}) \longrightarrow \mathcal{C}_{out} (A |_\mathcal{S}) $ is also bijective, where $ \mathcal{S} = \bigcup_{i=1}^g \mathcal{S}_i $. Therefore, we conclude that
$$ | \mathcal{C}_{out} (A |_\mathcal{S}) | = | \mathcal{C}_{out} (A |_\mathcal{R}) | = | \mathcal{C}_{out} | = | \mathcal{C}_{glob} | . $$
However, $ \mathcal{C}_{out} (A |_\mathcal{S}) \subseteq \mathbb{F}^{| \mathcal{S} |} $ and $ | \mathcal{S} | = \sum_{i=1}^g s_i = \sum_{i=1}^g {\rm Rk}(A_i|_{\mathcal{R}_i}) < k $, which is absurd since $ | \mathcal{C}_{glob} | = | \mathbb{F} |^k $.

2) Since $ \sum_{i=1}^g {\rm Rk}(A_i|_{\mathcal{R}_i}) \geq k $ and $ \mathcal{C}_{out} $ is an MSRD code, the erasure pattern can be corrected by Corollary \ref{cor sum-rank erasure correction}.
\end{proof}

This motivates the following definition.

\begin{definition} [\textbf{General MR-LRCs}] \label{def general MR LRC}
With notation as in Theorem \ref{th optimal arbitrary local codes}, we say that $ \mathcal{C}_{glob} $ is maximally recoverable (MR) for $ (\Gamma_i,\mathcal{C}_i)_{i=1}^g $ if it can correct all erasure paterns $ \mathcal{E} \subseteq [n] $ such that $ \sum_{i=1}^g {\rm Rk}(A_i|_{\mathcal{R}_i}) \geq k $, where $ \mathcal{E}_i = \mathcal{E} \cap \Gamma_i $ and $ \mathcal{R}_i = \Gamma_i \setminus \mathcal{E}_i $, for $ i = 1,2, \ldots, g $.
\end{definition}

We now show that this definition extends Definition \ref{def MR}.

\begin{corollary}
Let the notation be as in Theorem \ref{th optimal arbitrary local codes}, and assume that $ \mathcal{C}_i $ is an $ (r_i+\delta_i-1, r_i) $ MDS code, for $ i = 1,2, \ldots, g $. The following are equivalent:
\begin{enumerate}
\item
The code $ \mathcal{C}_{glob} $ is an MR-LRC for its $ (\Gamma_i, r_i ,\delta_i)_{i=1}^g $-localities according to Definition \ref{def MR}.
\item
The code $ \mathcal{C}_{glob} $ is an MR-LRC for $ (\Gamma_i,\mathcal{C}_i)_{i=1}^g $ according to Definition \ref{def general MR LRC}.
\end{enumerate} 
\end{corollary}
\begin{proof}
It follows from the definitions and the fact that, if $ A_i \in \mathbb{F}_q^{r_i \times (r_i+\delta_i-1)} $ generates $ \mathcal{C}_i $, then
$$ \Rk(A_i|_{\mathcal{R}_i}) = \min \{ r_i, | \mathcal{R}_i | \}, $$
since $ \mathcal{C}_i $ is MDS, for $ i = 1,2, \ldots, g $.
\end{proof}

We also deduce the following result, which proves Theorem \ref{th constr 1 is MR}.

\begin{corollary} \label{cor MSRD equivalent to MR LRC}
With notation as in Theorem \ref{th optimal arbitrary local codes}, the following are equivalent:
\begin{enumerate}
\item
$ \mathcal{C}_{out} \subseteq \mathbb{F}^N $ is MSRD over $ \mathbb{F}_q $ for the sum-rank length partition $ N = \sum_{i=1}^g r_i $.
\item
For all full-rank matrices $ A_i \in \mathbb{F}_q^{r_i \times n_i} $, for $ i = 1,2, \ldots, g $, the code $ \mathcal{C}_{glob} = \mathcal{C}_{out} \diag(A_1, A_2, \ldots, A_g) $ is an MR-LRC for $ (\Gamma_i,\mathcal{C}_i)_{i=1}^g $, being $ \mathcal{C}_i \subseteq \mathbb{F}^{n_i} $ the linear code generated by $ A_i $.
\end{enumerate}
In particular, Construction \ref{construction 1} gives MR-LRCs for arbitrary local linear codes with global fields $ \mathbb{F} = \mathbb{F}_{q^m} $, where $ q > g $ and $ m \geq \max_i r_i $, being $ r_i = \dim(\mathcal{C}_i) $, for $ i = 1,2, \ldots, g $. 
\end{corollary}
\begin{proof}
It follows by combining Corollary \ref{cor sum-rank erasure correction}, Corollary \ref{cor sum rank singleton} and Theorem \ref{th optimal arbitrary local codes}.
\end{proof}

We now show that $ m=r $ is the smallest extension degree of $ \mathbb{F} $ over $ \mathbb{F}_q $ that allows arbitrary local linear codes with localities up to $ r $, which is achieved by Construction \ref{construction 1}.

\begin{proposition} \label{prop minimum m for universal LRC}
For the positive integers $ g $ and $ r $ and the field $ \mathbb{F}_q $, the following are equivalent:
\begin{enumerate}
\item
$ \mathbb{F} = \mathbb{F}_{q^m} $ with $ m \geq r $.
\item
There exists a $ (gr,k) $ MSRD code $ \mathcal{C}_{out} \subsetneqq \mathbb{F}^{gr} $ over $ \mathbb{F}_q $, with $ k < gr $, for the sum-rank length partition $ gr = \sum_{i=1}^g r $.
\item
For all $ 1 \leq r_i \leq r $, $ i = 1,2, \ldots, g $, there exists an $ (N,k) $ MSRD code $ \mathcal{C}_{out} \subsetneqq \mathbb{F}^N $ over $ \mathbb{F}_q $, with $ k < \sum_{i=1}^g r_i $, for the sum-rank length partition $ N = \sum_{i=1}^g r_i $.
\end{enumerate}
\end{proposition}
\begin{proof}
Immediate from Corollary \ref{cor minimum m for MSRD} and the fact that linearized Reed-Solomon codes are MSRD and exist for the considered parameters (Definition \ref{def lin RS codes} and Proposition \ref{prop linRS are MSRD}).
\end{proof}

Observe that:

1) For $ h = 0 $, i.e., $ k = N = \sum_{i=1}^g r_i $, the whole space $ \mathcal{C}_{out} = \mathbb{F}^N $ is MSRD over any subfield (note that Corollary \ref{cor minimum m for MSRD} does not apply since $ \dd_{SR}(\mathbb{F}^N) = 1 $), hence we may take $ \mathbb{F} = \mathbb{F}_q $. In this case, $ \mathcal{C}_{glob} = \mathcal{C}_1 \times \mathcal{C}_2 \times \cdots \times \mathcal{C}_g $, and we recover Cartesian products. Furthermore, we may take $ q $ as the minimum common power of $ q_1, q_2, \ldots, q_g $, where $ \mathbb{F}_{q_i} $ is the local field for $ \mathcal{C}_i $, for $ i = 1,2, \ldots, g $, in accordance with the general construction.

2) For $ r = 1 $, we may take $ m=1 $, thus $ \mathbb{F} = \mathbb{F}_q $ again. Since the sum-rank metric for $ r_1 = r_2 = \ldots = r_g = 1 $ coincides with the Hamming metric, $ \mathcal{C}_{out} $ only needs to be MDS, and the corresponding linearized Reed-Solomon codes are classical Reed-Solomon codes. Hence we recover MDS global codes with local replication ($ \dim(\mathcal{C}_i) = 1 $, for $ i = 1,2, \ldots, g $).

An example of parameters for an MR-LRC with an MSRD outer code and different local linear codes can be found in Example \ref{example intro} in Subsection \ref{subsec partitioning initial localities}.

\subsection{Global Distances and Thresholds for Erasure Correction}  \label{subsec global distance arbitrary local codes}

In Theorem \ref{th optimal arbitrary local codes} we showed that MSRD outer codes correct all erasure patterns that are information-theoretically correctable for arbitrary disjoint local linear codes, and gave a description of such patterns. However, it is usual in the LRC literature to give the minimum distance of the global code, although its optimality is in general weaker than the MR condition. In this subsection, we give a formula for such global distances. It also shows the optimality of the global distance of Construction \ref{construction 1} even without assuming that $ r_1 \leq r_2 \leq \ldots \leq r_g $ and $ \delta_1 \geq \delta_2 \geq \ldots \geq \delta_g $, in contrast with \cite{chen-hao, kim}, and hence in contrast with all previous studies. 

Fix a full-rank matrix $ A_i \in \mathbb{F}_q^{r_i \times n_i} $ with $ 1 \leq r_i \leq n_i $, and let $ \mathcal{C}_i \subseteq \mathbb{F}^{n_i} $ be the linear code generated by $ A_i $, for $ i = 1, $ $2, $ $ \ldots, $ $ g $. Define $ n = n_1 + n_2 + \cdots + n_g $. For $ k = 1, $ $2, $ $ \ldots, $ $ \sum_{i=1}^g r_i $, define
\begin{equation}
\begin{split}
e(A,k) = \max \{ & e \in [n] \mid \min \{ \Rk (A|_\mathcal{R}) \mid \\
 & | \mathcal{R} | = n - e \} \geq k \},
\end{split}
\label{eq def of e_max}
\end{equation}
where $ A = \diag(A_1, A_2, \ldots, A_g) \in \mathbb{F}_q^{N \times n} $.

\begin{theorem} \label{th global distance and threshold}
Fix an $ (n,k) $ and $ (\Gamma_i, \mathcal{C}_i)_{i=1}^g $-LRC $ \mathcal{C}_{glob} \subseteq \mathbb{F}^n $ as in Definition \ref{def locally linear LRC}. The following hold:
\begin{enumerate}
\item
$ \dd_H(\mathcal{C}_{glob}) \leq e(A,k) + 1 $, for any outer code $ \mathcal{C}_{out} \subseteq \mathbb{F}^N $.
\item
$ \dd_H(\mathcal{C}_{glob}) = e(A,k) + 1 $ if $ \mathcal{C}_{glob} $ is an MR-LRC for the given $ (\Gamma_i, \mathcal{C}_i)_{i=1}^g $-localities.
\item
$ \dd_H(\mathcal{C}_{glob}) = e(A,k) + 1 $ if the outer code $ \mathcal{C}_{out} \subseteq \mathbb{F}^N $ is MSRD over $ \mathbb{F}_q $ for the sum-rank length partition $ N = \sum_{i=1}^g r_i $.
\end{enumerate}
\end{theorem}
\begin{proof}
It follows from Theorem \ref{th optimal arbitrary local codes} after unfolding the definitions.
\end{proof}

Observe that, if $ r_i = n_i = 1 $ (or simply $ r_i = n_i $ in general), for $ i = 1, $ $ 2, $ $ \ldots, g $, then 
$$ e(A,k) = n - k. $$
Therefore, the previous theorem recovers the classical Singleton bound and definition of MDS codes \cite{singleton}. In this case, optimal global distance is equivalent to MR due to the lack of linear redundancies in the matrices $ A_i $ (in other words, $ n_i - r_i = 0 $), for $ i = 1,2, \ldots, g $.

A bit more generally, we may give a simple formula for $ e(A,k) $ when the local codes are MDS and $ r_1 \leq r_2 \leq \ldots \leq r_g $ and $ \delta_1 \geq \delta_2 \geq \ldots \geq \delta_g $, which coincides with \cite[Th. 2]{chen-hao} and \cite[Th. 2]{kim} for disjoint local groups. It also recovers \cite[Th. 2]{kadhe} and \cite[Th. 8]{zeh-multiple} when $ \delta_1 = \delta_2 = \ldots = \delta_g = 2 $ for disjoint local groups. Finally, it recovers \cite[Th. 2.1]{kamath} and \cite[Eq. (2)]{gopalan} for equal localities and disjoint local groups.

\begin{proposition}
Assume that $ \mathcal{C}_i $ is MDS, for $ i = 1,2, \ldots, g $, $ r_1 \leq r_2 \leq \ldots \leq r_g $ and $ \delta_1 \geq \delta_2 \geq \ldots \geq \delta_g $, and let $ k = 1,2, \ldots, \sum_{i=1}^g r_i $. Let $ \ell = 0,1,2,\ldots, g-1 $ be the unique integer such that
\begin{equation}
\sum_{i=1}^\ell r_i < k \leq \sum_{i=1}^{\ell + 1} r_i.
\label{eq def of ell}
\end{equation}
Then it holds that
\begin{equation}
e(A,k) = n - k - \sum_{i=1}^\ell (\delta_i - 1).
\label{eq general bound}
\end{equation}
\end{proposition}
\begin{proof}
We present a sketch of the proof. Since $ \mathcal{C}_i $ is MDS, we have that
$$ {\rm Rk}(A_i |_{\mathcal{R}_i}) = \min \{ r_i, | \mathcal{R}_i | \} = \min \{ r_i, r_i + \delta_i - 1 - |\mathcal{E}_i| \}, $$
for any $ \mathcal{R}_i \subseteq [n_i] $, for $ i = 1,2, \ldots, g $. Hence, for a given $ e \in [n] $, the worst-case number of erasures is
\begin{equation*}
\begin{split}
\rho = N - \min \{ & \sum_{i=1}^g \min \{ r_i, r_i + \delta_i - 1 - e_i \} \mid \\
& \sum_{i=1}^g e_i = e, 0 \leq e_i \leq r_i + \delta_i - 1, 1 \leq i \leq g \} .
\end{split}
\end{equation*}
As argued in the proof of \cite[Th. 24]{rawat}, the worst-case erasure pattern is achieved when erasures concentrate in the smallest number of local groups. In our case, these are the last groups since $ r_1 \leq r_2 \leq \ldots \leq r_g $ and $ \delta_1 \geq \delta_2 \geq \ldots \geq \delta_g $. Let $ 0 < \Delta \leq r_{\ell+1} $ be such that $ k = \sum_{i=1}^\ell r_i + \Delta $. If $ e $ is
\begin{equation*}
\begin{split}
e & = \sum_{j=1}^{g-\ell-1} (r_{g-j+1} + \delta_{g-j+1} - 1) + (r_{\ell+1} + \delta_{\ell+1} - 1 - \Delta) \\
 & = N - k + \sum_{j=1}^{g-\ell} (\delta_{g-j+1} - 1) = n - k - \sum_{i=1}^\ell (\delta_i - 1),
\end{split}
\end{equation*}
then it holds that $ \rho = N - k $. Hence more than $ e $ erasures will not be correctable. Thus $ e(A,k) = e $ and we are done.
\end{proof}

\section{Universal and Dynamic Properties} \label{sec dynamic properties}

In this section, we show how to perform local recodings (Subsection \ref{subsec arbitrary local recodings}), partition localities (Subsection \ref{subsec partitioning initial localities}), obtain multi-layer or hierarchical MR-LRCs (Subsection \ref{subsec hierarchical MR-LRCs}), and change the initial localities, number of local groups and file components (Subsection \ref{subsec changes of initial localities}), when using Construction \ref{construction 1}.

\subsection{Arbitrary and Efficient Local Linear Recodings} \label{subsec arbitrary local recodings}

We now show that the architecture described in (\ref{eq architecture arbitrary local codes}) (Lemma \ref{lemma equivalence linear local codes}) enables any local linear recoding. In other words, the local linear codes can be changed to any other local linear codes by only performing linear operations, inside each local group, over the local fields. The outer code remains unchanged, thus in case it is MSRD, the MR condition is preserved by Corollary \ref{cor MSRD equivalent to MR LRC}, and there is no need for global recoding.

We start by introducing local recoding matrices.

\begin{definition} [\textbf{Local recoding matrices}]
Let $ \mathcal{C}_{loc}^{(i)} \subseteq \mathbb{F}_q^{n_i} $ and $ \mathcal{D}_{loc}^{(i)} \subseteq \mathbb{F}_q^{n_i^\prime} $ be local $ \mathbb{F}_q $-linear codes with full-rank generator matrices $ A_i \in \mathbb{F}_{q_i}^{r_i \times n_i} $ and $ B_i \in \mathbb{F}_{q_i^\prime}^{r_i \times n_i^\prime} $, respectively, with $ q $ a power of $ q_i $ and $ q_i^\prime $, for $ i = 1,2, \ldots, g $. We define the corresponding local recoding matrices as the unique rank-$ r_i $ matrices $ T_i \in \mathbb{F}_q^{n_i \times n_i^\prime} $ such that $ B_i = A_i T_i $, for $ i = 1,2, \ldots, g $. 
\end{definition}

The existence of such recoding matrices is straightforward by linear algebra: For $ i = 1,2, \ldots, g $, there exists $ C_i \in \mathbb{F}_{q_i}^{r_i \times n_i} $ such that $ A_i C_i^T = I_{r_i} $, since $ \Rk (A_i) = r_i $. Hence
\begin{equation}
T_i = C_i^T B_i \in \mathbb{F}_q^{n_i \times n_i^\prime} .
\label{eq def local recoding matrices}
\end{equation}

Now let $ \mathcal{C}_{out} \subseteq \mathbb{F}^N $ be an outer MSRD code, with $ N = \sum_{i=1}^g r_i $. The corresponding global codes in (\ref{eq architecture arbitrary local codes}) are given by
\begin{equation*}
\begin{split}
\mathcal{C}_{glob} & = \mathcal{C}_{out} \diag(A_1, A_2, \ldots, A_g) \subseteq \mathbb{F}^n, \\
\mathcal{D}_{glob} & = \mathcal{C}_{out} \diag(B_1, B_2, \ldots, B_g) \subseteq \mathbb{F}^{n^\prime},
\end{split}
\end{equation*}
respectively, where $ n = \sum_{i=1}^g n_i $ and $ n^\prime = \sum_{i=1}^g n_i^\prime $. Therefore, it holds that
\begin{equation}
\mathcal{D}_{glob} = \mathcal{C}_{glob} \diag(T_1, T_2, \ldots, T_g).
\label{eq local recoding}
\end{equation}
This block-diagonal matrix multiplication can be understood as the local groups recoding their local stored data over the local fields, without need of communication between groups, global recoding or change of the outer code. 

The complexity of recoding the $ i $th local group is as follows: First, decoding the initial code has, in general, complexity $ \mathcal{O}( r_i^3 \log(q_i)^2 ) $ over $ \mathbb{F}_2 $. Second, encoding with the new code has, in general, complexity $ \mathcal{O}( r_i^2 \log(q_i^\prime)^2 ) $ over $ \mathbb{F}_2 $. Observe that global recoding in $ \mathbb{F}_{q^m}^{gr} $ has, in general, complexity $ \mathcal{O}( r^3 g^3 m^2 \log(q)^2 ) $ over $ \mathbb{F}_2 $, where $ m \geq r $ and $ g \gg r $.

An example of such local linear recodings can be found in Example \ref{example intro} below.

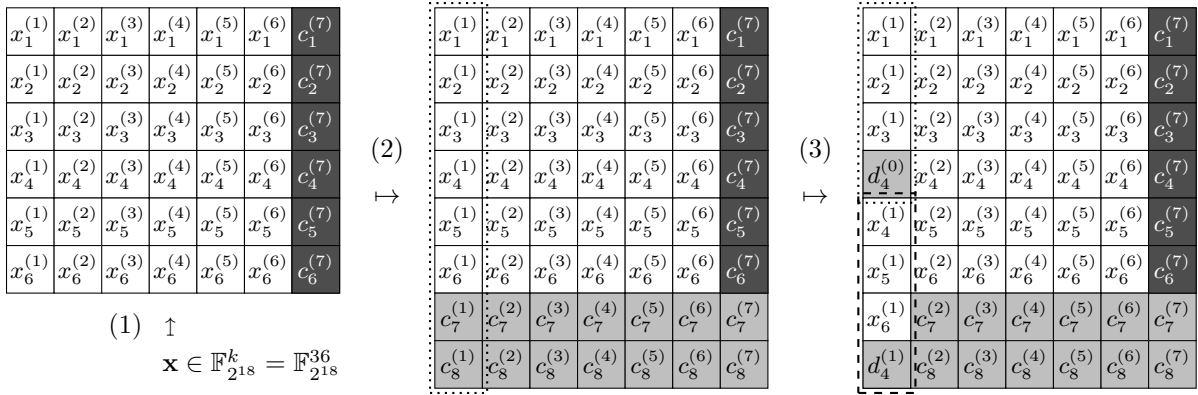
\begin{figure*} [!t]
\begin{center}
\begin{tabular}{c@{\extracolsep{1cm}}c}
\begin{tikzpicture}[
square/.style = {draw, rectangle, 
                 minimum size=\m, outer sep=0, inner sep=0, font=\small,
                 },
                        ]
\def\m{18pt}
\def\w{7}
\def\h{8}
\def\loc{6}
    \pgfmathsetmacro\uw{int(\w/2)}
    \pgfmathsetmacro\uh{int(\h/2)}

   \node [] at (5.65*\m,-8*\m) {$ \mathbf{x} \in \mathbb{F}_{2^{18}}^k = \mathbb{F}_{2^{18}}^{36} $}; 
   \node [] at (3*\m,-7.2*\m) {$ (1) $};   
   \node [] at (4*\m,-7.2*\m) {$ \upmapsto $};

  \foreach \x in {1,...,\w}
    \foreach \y in {1,...,\loc}
       {    
           \node [square, fill=white]  (\x,\y) at (\x*\m,-\y*\m) {$ x^{(\x)}_{\y} $};
       }
       
  \foreach \y in {1,...,\loc}
       {
           \node [square, fill=black!70]  (\w,\y) at (\w*\m,-\y*\m) {{\color{white}$ c^{(\w)}_{\y} $}};
       }

\node [] (\w + 1, \h) at (\w*\m + 1.5*\m, -3.5*\m) {$ (2) $};
\node [] (\w + 1, \h) at (\w*\m + 1.5*\m, -4.5*\m) {$ \mapsto $};

\def\i{9}
  \foreach \x in {1,...,\w}
    \foreach \y in {1,...,\h}
       {    
           \ifnum\y>\loc
               \node [square, fill=gray!50]  (\x,\y) at (\x*\m + \i*\m,-\y*\m) {$ c^{(\x)}_{\y} $};
           \else
               \node [square, fill=white]  (\x,\y) at (\x*\m + \i*\m,-\y*\m) {$ x^{(\x)}_{\y} $};
           \fi
       }
       
  \foreach \y in {1,...,\loc}
       {
           \node [square, fill=black!70]  (\w,\y) at (\w*\m + \i*\m,-\y*\m) {{\color{white}$ c^{(\w)}_{\y} $}};
       }
       
   \draw[thick, dotted] (0.4*\m + \i*\m, -0.4*\m) rectangle (1.6*\m + \i*\m,-8.6*\m); 

\node [] (\w + 1, \h) at (\w*\m + 10.5*\m, -3.5*\m) {$ (3) $};
\node [] (\w + 1, \h) at (\w*\m + 10.5*\m, -4.5*\m) {$ \mapsto $};

\def\i{18}
  \foreach \x in {1,...,\w}
    \foreach \y in {1,...,\h}
       {    
           \ifnum\y>\loc
               \node [square, fill=gray!50]  (\x,\y) at (\x*\m + \i*\m,-\y*\m) {$ c^{(\x)}_{\y} $};
           \else
               \node [square, fill=white]  (\x,\y) at (\x*\m + \i*\m,-\y*\m) {$ x^{(\x)}_{\y} $};
           \fi
       }

  \foreach \y in {1,...,\loc}
       {
           \node [square, fill=black!70]  (\w,\y) at (\w*\m + \i*\m,-\y*\m) {{\color{white}$ c^{(\w)}_{\y} $}};
       }

   \node [square, fill=gray!50] at (\m + \i*\m,-4*\m) {$ d^{(0)}_4 $}; 
   \node [square, fill=white] at (\m + \i*\m,-5*\m) {$ x^{(1)}_4 $}; 
   \node [square, fill=white] at (\m + \i*\m,-6*\m) {$ x^{(1)}_5 $}; 
   \node [square, fill=white] at (\m + \i*\m,-7*\m) {$ x^{(1)}_6 $};
   \node [square, fill=gray!50] at (\m + \i*\m,-8*\m) {$ d^{(1)}_4 $};  
   
   \draw[thick, dotted] (0.4*\m + \i*\m, -0.4*\m) rectangle (1.6*\m + \i*\m,-4.6*\m); 
   \draw[thick, dashed] (0.4*\m + \i*\m, -4.4*\m) rectangle (1.6*\m + \i*\m,-8.6*\m); 
        
\end{tikzpicture}

\end{tabular}
\end{center}

\caption{Illustration of Example \ref{example intro}. The uncoded file is denoted by $ \mathbf{x} $ and consists, in this figure, of $ k = 36 $ symbols over $ \mathbb{F}_{2^{18}} $ (or $ \mathbb{F}_{2^{18}}^\alpha $ by folding). It is encoded systematically using an MSRD outer code over $ \mathbb{F}_{2^{18}} $ to produce $ h = gr-k = 6 $ global parities (depicted in dark grey). All symbols are then arranged in $ g=7 $ local groups (corresponding to columns), each of size $ r = 6 $. This first process is denoted by (1). Finally, we compute two local parities (depicted in lighter grey) for each local group using an $ (8,6) $ MDS code over $ \mathbb{F}_{2^3} $ column-wise (hence $ g=7 $ times). This second process is denoted by (2). As explained in Example \ref{example intro}, we may recode locally the first group using two $ (4,3) $ MDS codes over $ \mathbb{F}_2 $, to generate the new local parities $ d^{(0)}_4 $ and $ d^{(1)}_4 $, while the underlying outer code and the rest of local groups remain unchanged. This recoding is denoted by (3). }
\label{fig partitioning localities}
\end{figure*}

\begin{example} \label{example intro}
Let $ g = 7 $, $ r = 6 $, and $ 1 \leq k \leq 42 $. With Construction \ref{construction 1}, the global field is $ \mathbb{F}_{2^{18}} $ and local fields are $ \mathbb{F}_{2^3} $. We first encode each block of $ k $ symbols (over $ \mathbb{F}_{2^{18}} $) of the file with an MSRD code of length $ gr = 42 $ to obtain an outer codeword. Choose now any seven $ (8,6) $ MDS codes over $ \mathbb{F}_{2^3} $ for the seven groups and recode the outer codeword with their Cartesian product. We then obtain an MR-LRC with $ 7 $ groups, each with locality $ 6 $ over $ \mathbb{F}_{2^3} $, allowing fast local repair. By the MR property, the code can correct any $ h = 42-k $ more erasures than the simple Cartesian product of the MDS codes.

Imagine that the data in the first group becomes hot data. We may partition that group into two subgroups, and recode the corresponding block of the outer codeword with two $ (4,3) $ MDS codes over $ \mathbb{F}_2 $. This allows very fast local repair by only XORing, at the expense of lower local distance (only in that group). Now we have $ 6 $ local groups with locality $ 6 $ over $ \mathbb{F}_{2^3} $, and $ 2 $ local groups with locality $ 3 $ over $ \mathbb{F}_2 $. 

The transition only requires turning an $ (8,6) $ MDS code over $ \mathbb{F}_{2^3} $ into the Cartesian product of two $ (4,3) $ MDS codes over $ \mathbb{F}_2 $, which can be performed efficiently, compared to global recoding of all $ 7 $ groups of length $ 8 $ over $ \mathbb{F}_{2^{18}} $. We may equally return to the original global code, which remains MR-LRC in both settings during the whole process. See Figure \ref{fig partitioning localities} for an illustration when $ k=36 $ (thus $ h=6 $).

Observe that Gabidulin-based LRCs would require the global field $ \mathbb{F}_{2^{3 \times 42}} = \mathbb{F}_{2^{126}} $, and \cite{gabrys} might improve our global field size only if $ h = 42 - k = 1,2,3,4,5 $. The strengths of our approach become clearer as $ g $ grows, while $ r $ remains constant, see Example \ref{example second intro}.
\end{example}

\subsection{Partitioning Local Groups and Initial Localities} \label{subsec partitioning initial localities}

As a consequence of the local recodings in the previous subsection, we show now how to partition localities.

The main observation is that codes that are MSRD for a given sum-rank length partition are also MSRD for finer partitions.

\begin{theorem} \label{th sum-rank partitions}
For $ i = 1,2, \ldots, g $, partition $ r_i = \sum_{j=1}^{g_i} r_{i,j} $. Denote by $ \dd_{SR} $ and $ \dd^{ref}_{SR} $, the sum-rank metrics in $ \mathbb{F}_{q^m}^N $ for the sum-rank length partitions $ N = \sum_{i=1}^g r_i $ and $ N = \sum_{i=1}^g \sum_{j=1}^{g_i} r_{i,j} $, respectively. For a code $ \mathcal{C} \subseteq \mathbb{F}_{q^m}^N $ (linear or non-linear), it holds that
\begin{equation}
\dd_{SR}(\mathcal{C}) \leq \dd^{ref}_{SR}(\mathcal{C}).
\label{eq sum-rank partitions}
\end{equation}
In particular, if $ \mathcal{C} $ is MSRD with respect to $ \dd_{SR} $, then it is MSRD with respect to $ \dd^{ref}_{SR} $.
\end{theorem}
\begin{proof}
Immediate from Theorem \ref{th sum-rank distance is min among hamming distances} and Corollary \ref{cor sum rank singleton}.
\end{proof}

Note that, when $ g = 1 $ and $ g_1 = N $, this theorem recovers the well-known fact that $ \dd_R(\mathcal{C}) \leq \dd_H(\mathcal{C}) $, where $ \dd_R $ and $ \dd_H $ denote rank and Hamming distances, respectively. See, for instance, \cite{gabidulin}.

Hence, using Corollary \ref{cor MSRD equivalent to MR LRC}, we deduce the following result on partitioning local groups in MR-LRCs.

\begin{corollary} \label{cor partitioning localities}
Let $ \mathcal{C}_{out} \subseteq \mathbb{F}_{q^m}^N $ be MSRD for the sum-rank length partition $ N = \sum_{i=1}^g r_i $, with $ m \geq \max_i r_i $. For $ i = 1,2, \ldots, g $, partition $ r_i = \sum_{j=1}^{g_i} r_{i,j} $, and let $ A_{i,j} \in \mathbb{F}_{q_{i,j}}^{r_{i,j} \times n_{i,j}} $ be full-rank generator matrices of codes $ \mathcal{C}_{loc}^{(i,j)} \subseteq \mathbb{F}_{q_{i,j}}^{n_{i,j}} $, with $ 1 \leq r_{i,j} \leq n_{i,j} $ and $ q $ a common power of $ q_{i,j} $, for $ j = 1,2, \ldots, g_i $ and $ i = 1,2, \ldots, g $. The code
$$ \mathcal{C}_{glob} = \mathcal{C}_{out} \diag(A_1, A_2, \ldots, A_g) \subseteq \mathbb{F}_{q^m}^n, $$
where $ A_i = \diag (A_{i,1}, A_{i,2}, \ldots, A_{i,g_i}) \in \mathbb{F}_q^{r_i \times n_i} $, for $ i = 1,2, \ldots, g $, is an MR-LRC for the local codes $ ((\mathcal{C}_{loc}^{(i,j)})_{j=1}^{g_i})_{i=1}^g $.
\end{corollary}

Observe that such partitionings can be performed efficiently by local linear recoding as in the previous subsection. An example of such partitioning, for only the first initial locality $ r_1 = \sum_{j=1}^{g_1} r_{1,j} $, can be found in Example \ref{example intro} (see also Fig. \ref{fig partitioning localities}).

\subsection{Multi-layer or hierarchical  MR-LRCs} \label{subsec hierarchical MR-LRCs}

In this subsection, we introduce and show how to obtain hierarchical MR-LRCs. Codes with hierarchical localities were introduced in \cite{hierarchical}, and hierarchical MR-LRCs have been introduced independently in the parallel work \cite{nair}. Note however that \cite{hierarchical, nair} consider equal localities and local distances and two-level hierarchies, whereas we consider the general case (see Definition \ref{def hierarchical}). 

We start by noting that partitioning localities as in Corollary \ref{cor partitioning localities} is simply using Cartesian products as local codes. Cartesian products are precisely MR-LRCs for very high information rates (no global parities). Exactly as in the previous corollary, instead of choosing $ A_i = \diag (A_{i,1}, A_{i,2}, \ldots, A_{i,g_i}) $, we may choose $ A_i $ to be the generator matrix of any MR-LRC. 

A precise general definition can be given as follows. To this end, we will need the notion of a \textit{rooted tree}, meaning a connected finite graph with no cycles where a particular vertex is given the name \textit{root}. The \textit{leaves} of a rooted tree are those vertices of the tree with no children (see Fig. \ref{fig hierarchical graph}).

\begin{definition} [\textbf{Hierarchical MR-LRCs}] \label{def hierarchical}
We define linear hierarchical MR-LRCs recursively on the family of rooted trees as follows. First, let $ \mathcal{G} $ be a rooted tree formed by a single vertex and let $ \mathcal{C} $ be a $ (k+h,k) $ MDS code. We say that $ \mathcal{C} $ is a hierarchical MR-LRC with parameters $ (\mathcal{G}, k, h) $.

Let $ \mathcal{G} $ be a rooted tree with root $ v_0 $. Let $ v_1, v_2, \ldots, v_g $ be the children of $ v_0 $, and let $ \mathcal{G}_i $ be the rooted subtree of $ \mathcal{G} $ with root $ v_i $ formed by $ v_i $ and all its descendants, for $ i = 1,2, \ldots, g $. Let also $ l_1, l_2, \ldots, l_L $ be the leaves of $ \mathcal{G} $, and let $ r_j \geq 1 $ be positive integers, for $ j = 1,2, \ldots, L $. Let $ h_v \geq 1 $ be positive integers, for $ v \in \mathcal{V} $, where $ \mathcal{V} $ is the vertex set of $ \mathcal{G} $. Let $ \mathcal{C}_{glob} \subseteq \mathbb{F}^n $ be a linear code satisfying the conditions in Lemma \ref{lemma equivalence linear local codes}. We say that $ \mathcal{C}_{glob} $ is a hierarchical MR-LRC with parameters $ (\mathcal{G}, (r_j)_{j=1}^L, (h_v)_{v \in \mathcal{V}}) $ if:
\begin{enumerate}
\item
$ \mathcal{C}_{glob} $ is an MR-LRC for $ (\Gamma_i,\mathcal{C}_i)_{i=1}^g $, where $ \Gamma_i $ and $ \mathcal{C}_i $ are as in Lemma \ref{lemma equivalence linear local codes}, for $ i = 1,2, \ldots, g $, 
\item
$ h_{v_0} $ is the number of global parities of $ \mathcal{C}_{glob} $ (see Definition \ref{def global and local parities}), and
\item
the local code $ \mathcal{C}_i \subseteq \mathbb{F}^{n_i} $ is itself a hierarchical MR-LRC with parameters $ (\mathcal{G}_i, (r_j)_{j \in J_i}, (h_v)_{v \in \mathcal{V}_i}) $, for $ i = 1,2, \ldots, g $.
\end{enumerate}
Here, $ J_i $ is the set of indices $ j $ such that $ l_j $ is a leaf of $ \mathcal{G}_i $, and $ \mathcal{V}_i $ is the vertex set of $ \mathcal{G}_i $, for $ i = 1,2, \ldots, g $.
\end{definition}

\begin{figure} [!h]
\begin{center}
	
	\begin{tikzpicture}[line width=1pt, scale=1]
		\tikzstyle{every node}=[inner sep=0pt, minimum width=4.5pt]
		
		\draw (0,3) node (v0) [draw, circle, fill=gray] {};
		\draw (-1,2) node (l1) [draw, circle, fill=gray] {};		
		\draw (1,2) node (v1) [draw, circle, fill=gray] {};		
		\draw (0,1) node (v2) [draw, circle, fill=gray] {};			
		\draw (2,1) node (l4) [draw, circle, fill=gray] {};	
		\draw (-1,0) node (l2) [draw, circle, fill=gray] {};	
		\draw (1,0) node (l3) [draw, circle, fill=gray] {};	
				
		\draw[middlearrow={stealth}, ultra thick] (v0) -- (l1);		
		\draw[middlearrow={stealth}, ultra thick] (v0) -- (v1);
		\draw[middlearrow={stealth}, ultra thick] (v1) -- (v2);
		\draw[middlearrow={stealth}, ultra thick] (v2) -- (l2);		
		\draw[middlearrow={stealth}, ultra thick] (v2) -- (l3);
		\draw[middlearrow={stealth}, ultra thick] (v1) -- (l4);
		
		\draw (0,3.5) node [] {$ v_0 $};
		\draw (-1.8,2) node [] {$ v_1 = l_1 $};	
		\draw (1.5,2) node [] {$ v_2 $};	
		\draw (-0.5,1) node [] {$ v_3 $};	
		\draw (2,0.5) node [] {$ l_4 $};
		\draw (-1,-0.5) node [] {$ l_2 $};	
		\draw (1,-0.5) node [] {$ l_3 $};

	\end{tikzpicture}	
	
\end{center}

\caption{The rooted tree, with notation as in Definition \ref{def hierarchical}, corresponding to the hierarchical MR-LRC depicted in Fig. \ref{fig hierarchical} below. Here, $ g = 2 $ and $ L = 4 $. }
\label{fig hierarchical graph}
\end{figure}
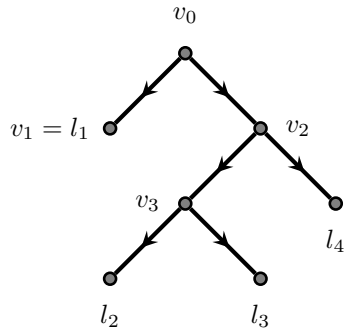

\begin{figure} [!h]
\begin{center}
	
\begin{tikzpicture}[
square/.style = {draw, rectangle, 
                 minimum size=\m, outer sep=0, inner sep=0, font=\small,
                 },
                        ]
\def\m{18pt}
\def\w{7}
\def\h{6}
\def\loc{4}
    \pgfmathsetmacro\uw{int(\w/2)}
    \pgfmathsetmacro\uh{int(\h/2)}

\def\i{10}
  \foreach \x in {1,...,\w}
    \foreach \y in {1,...,\h}
       {    
            \node [square, fill=white]  (\x,\y) at (\x*\m + \i*\m,-\y*\m) {$  $};
       }
       
  \foreach \y in {2,...,\loc}
       {
           \node [square, fill=black!70]  (\w,\y) at (\w*\m + \i*\m,-\y*\m) {{\color{white}$ v_2 $}};
       }
       
  \foreach \y in {3,...,\loc}
       {
           \node [square, fill=black!70]  (5,\y) at (5*\m + \i*\m,-\y*\m) {{\color{white}$ v_3 $}};
       }
       
    \foreach \y in {3,...,\loc}
       {
           \node [square, fill=black!70]  (3,\y) at (3*\m + \i*\m,-\y*\m) {{\color{white}$ v_0 $}};
       }

   \draw [decorate, thick,decoration={brace,amplitude=5pt}]
   (\i*\m + 1.5*\m, -\h*\m - \m) -- (\i*\m + 0.5*\m, -\h*\m - \m) node[midway,yshift=-1.5em]{$ v_1 = l_1 $}; 

   \draw [decorate, thick,decoration={brace,amplitude=5pt}]
   (\i*\m + 3.5*\m, -\h*\m - \m) -- (\i*\m + 1.5*\m, -\h*\m - \m) node[midway,yshift=-1.5em]{$ l_2 $}; 

   \draw [decorate, thick,decoration={brace,amplitude=5pt}]
   (\i*\m + 5.5*\m, -\h*\m - \m) -- (\i*\m + 3.5*\m, -\h*\m - \m) node[midway,yshift=-1.5em]{$ l_3 $}; 

   \draw [decorate, thick,decoration={brace,amplitude=5pt}]
   (\i*\m + 7.5*\m, -\h*\m - \m) -- (\i*\m + 5.5*\m, -\h*\m - \m) node[midway,yshift=-1.5em]{$ l_4 $};

\node [square, fill=gray!50]  at (1*\m + \i*\m,-4*\m) {$ l_1 $};
\node [square, fill=gray!50]  at (1*\m + \i*\m,-5*\m) {$ l_1 $};

\node [square, fill=gray!50]  at (2*\m + \i*\m,-5*\m) {$ l_2 $};
\node [square, fill=gray!50]  at (2*\m + \i*\m,-6*\m) {$ l_2 $};
\node [square, fill=gray!50]  at (3*\m + \i*\m,-5*\m) {$ l_2 $};

\node [square, fill=gray!50]  at (4*\m + \i*\m,-5*\m) {$ l_3 $};
\node [square, fill=gray!50]  at (4*\m + \i*\m,-6*\m) {$ l_3 $};
\node [square, fill=gray!50]  at (5*\m + \i*\m,-5*\m) {$ l_3 $};

\node [square, fill=gray!50]  at (6*\m + \i*\m,-5*\m) {$ l_4 $};
\node [square, fill=gray!50]  at (6*\m + \i*\m,-6*\m) {$ l_4 $};
\node [square, fill=gray!50]  at (7*\m + \i*\m,-4*\m) {$ l_4 $};

\node [square, fill=white]  at (1*\m + \i*\m,-6*\m) {$  $};
\pic [square, fill=white] at (1*\m + \i*\m,-6*\m) {lower={white}{$ $}};
\pic [square, fill=white] at (1*\m + \i*\m,-6*\m) {left={white}{$ $}};

\node [square, fill=white]  at (3*\m + \i*\m,-6*\m) {$  $};
\pic [square, fill=white] at (3*\m + \i*\m,-6*\m) {lower={white}{$ $}};

\node [square, fill=white]  at (5*\m + \i*\m,-6*\m) {$  $};
\pic [square, fill=white] at (5*\m + \i*\m,-6*\m) {lower={white}{$ $}};

\node [square, fill=white]  at (7*\m + \i*\m,-5*\m) {$  $};
\pic [square, fill=white] at (7*\m + \i*\m,-5*\m) {lower={white}{$ $}};
\pic [square, fill=white] at (7*\m + \i*\m,-5*\m) {right={white}{$ $}};

\node [square, fill=white]  at (7*\m + \i*\m,-6*\m) {$  $};
\pic [square, fill=white] at (7*\m + \i*\m,-6*\m) {upper={white}{$ $}};
\pic [square, fill=white] at (7*\m + \i*\m,-6*\m) {lower={white}{$ $}};
\pic [square, fill=white] at (7*\m + \i*\m,-6*\m) {right={white}{$ $}};
 
\draw[thick, dashed] (1.4*\m + \i*\m, -0.4*\m) rectangle (5.6*\m + \i*\m,-6.6*\m); 
\draw[thick, dotted] (1.3*\m + \i*\m, 0.6*\m) rectangle (7.7*\m + \i*\m,-6.7*\m); 
 
\node at (2.9*\m + \i*\m, 0*\m) {$ v_3 $}; 
\node at (3.4*\m + \i*\m, -0.1*\m) {\rotatebox[origin=c]{270}{$\Rsh$}};  
\node at (8.3*\m + \i*\m, 0*\m) {$ v_2 $};

\end{tikzpicture}	
	
\end{center}

\caption{Illustration of a hierarchical MR-LRC corresponding to the rooted tree in Fig. \ref{fig hierarchical graph}. White boxes correspond to information symbols and grey boxes correspond to parities. The vertex associated with the code producing the parity is written inside the box. Boxes inside the dashed box form codewords of the code corresponding to $ v_3 $, whereas boxes inside the dotted box form codewords of the code corresponding to $ v_2 $, and similarly for the leaves. }
\label{fig hierarchical}
\end{figure}
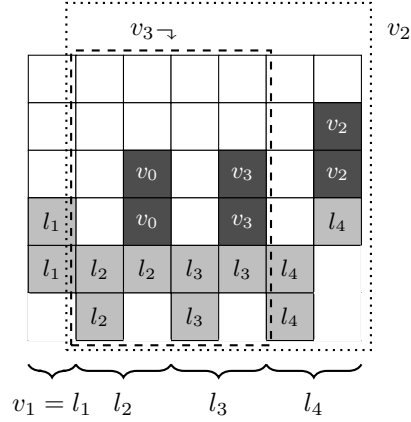

Observe that we only need to specify the rooted tree giving the hierarchy of the local codes, the global redundancy $ h_v $ at each vertex $ v $ of the tree, plus the locality $ r_j $ at each leaf $ l_j $. At a given vertex distinct from the root and leaves, the corresponding code is both a local and a global code. 

Note that the codes at the leaves are MDS by definition. If we denote by $ \delta_j = h_{l_j} + 1 $ the distance of the local code at the leaf $ l_j $, for $ j = 1,2, \ldots, L $, then the length and dimension of $ \mathcal{C}_{glob} $ are given, respectively, by 
\begin{equation}
n = \sum_{j=1}^L (r_j + \delta_j - 1) \quad \textrm{and} \quad k = n - \sum_{v \in \mathcal{V}} h_v.
\label{eq length dim hierarchical}
\end{equation}

Fig. \ref{fig hierarchical graph} shows a rooted tree, and Fig. \ref{fig hierarchical} depicts a hierarchical MR-LRC for such a tree, with a certain choice of localities and parities at each vertex.

Imposing that at each vertex of the tree the corresponding code is MR means that all possible information-theoretically correctable erasure patterns (for the given tree, localities and parities) can be corrected by the code at that vertex. 

The notion of hierarchical MR-LRC from \cite[Def. 6]{nair} is recovered from Definition \ref{def hierarchical} by choosing rooted trees where each leaf has depth two (the length of the path from the root to the leaf), and all localities and parities at vertices with the same depth are equal. This hierarchy is usually called \textit{two-level hierarchy}.

Observe that, if we drop the MR condition in Definition \ref{def hierarchical} (thus the MDS condition on trees with a single vertex), then we obtain general hierarchical LRCs. For two-level hierarchies, we recover the notion of hierarchical LRCs from \cite[Def. 2]{hierarchical}.

Finally, adapting recursively Construction \ref{construction 1} as in Corollary \ref{cor MSRD equivalent to MR LRC}, we may obtain explicit MR-LRCs for any choice of tree, localities and parities. The leaves are chosen as short MDS codes (e.g., Reed-Solomon) and at each vertex different from the leaves, we use a tailored linearized Reed-Solomon.

\begin{theorem} \label{th construction hierarchical}
Let $ \mathcal{G} $ be any rooted tree with vertex set $ \mathcal{V} $ and leaves $ l_1, l_2, \ldots, l_L \in \mathcal{V} $. Let $ r_j , h_v \geq 1 $ be arbitrary integers, for $ j = 1,2, \ldots, L $ and for $ v \in \mathcal{V} $. There exists a hierarchical MR-LRC $ \mathcal{C} \subseteq \mathbb{F}^n $ with parameters $ (\mathcal{G}, (r_j)_{j=1}^L, (h_v)_{v \in \mathcal{V}}) $. Its length and dimension are given as in (\ref{eq length dim hierarchical}). 

Let $ v_1, v_2, \ldots, v_g \in \mathcal{V} $ be the children of the root of $ \mathcal{G} $ and let the notation be as in Definition \ref{def hierarchical}. Then the global field of $ \mathcal{C} $ is $ \mathbb{F} = \mathbb{F}_{q^m} $, where $ q > g $ and $ m = \max \{ k_1, k_2, \ldots, k_g \} $, where $ k_i $ is the dimension of the code $ \mathcal{C}_i $ at vertex $ v_i $, which can be recursively computed using (\ref{eq length dim hierarchical}) for the subtree $ \mathcal{G}_i $, for $ i = 1,2, \ldots, g $.
\end{theorem}

Observe that Fig. \ref{fig hierarchical graph} can be used to describe a systematic encoding procedure for such a construction, as done in Fig. \ref{fig encoding for construction I}.

Consider the two-level hierarchy described above, where the root has children $ v_1, v_2, \ldots, v_g $ and each $ v_i $ has children $ l_{i,1}, l_{i,2}, \ldots, l_{i,t} $, for $ i = 1,2, \ldots, g $, hence $ L = gt $. We will assume equal localities and parities at each level. Let $ r $ and $ \delta $ be the locality and local distance at each leaf, let $ h_2 $ be the global parities at the chilren $ v_1, v_2, \ldots, v_g $, and let $ h_1 $ be the global parities at the root. By Theorem \ref{th construction hierarchical} and (\ref{eq length dim hierarchical}), the field sizes at the root, at its children and at the leaves are roughly
\begin{equation}
g^{tr - h_2}, \quad t^r \quad \textrm{and} \quad r+\delta - 1,
\label{eq field sizes two-level}
\end{equation}
respectively. In general, ``local'' decoding at lower layers is more efficient than ``global'' decoding at upper layers, whose erasure-correction algorithms are triggered less frequently.

\subsection{Recursive Encoding, and Changes of Initial Localities, File Components and Number of Local Groups} \label{subsec changes of initial localities}

In the previous subsection, we studied how to partition the initial localities without global recoding. In this subsection, we show how to modify, without global recoding, the initial localities up to $ m $ ($ \mathbb{F} = \mathbb{F}_{q^m} $), the initial file size $ k $, and the initial number of local groups $ g $ up to $ q-1 $. Note that the restriction $ k \leq \sum_{i=1}^g r_i $ must always hold by Corollary \ref{cor max k for initial localities}. Changes in these three parameters imply changing the number of global parities $ h = \sum_{i=1}^g r_i - k \geq 0 $, whereas changes in the local codes, as in the previous subsections, imply changing the number of local parities.

These processes are of interest when one desires to prepare iteratively a DSS that stores a final file up to a given size. One starts with a small number of local groups and a small file, and adds new groups, localities and file components over time. In such a scenario, it is desirable to encode the final file recursively, protecting intermediate files from erasures, without global recoding at each stage. Since these processes are reversible (the three parameters can be decreased), they can be used to remove and/or update file components over time without global recoding.

Let $ \mathcal{C}_k \subseteq \mathbb{F}_{q^m}^{(q-1)m} $ be a $ k $-dimensional linearized Reed-Solomon code (Definition \ref{def lin RS codes}), for $ k = 0,1,2, \ldots, N_0 = (q-1)m $, for the sum-rank length partition $ N_0 = \sum_{i=1}^{q-1} m $. Assume that $ (\mathcal{C}_k)_{k=0}^{(q-1)m} $ form a nested sequence of codes with nested generator matrices (placing extra rows at the end)
\begin{equation}
G_k = (G_{k,1} | G_{k,2} | \ldots | G_{k,q-1}) \in \mathbb{F}_{q^m}^{k \times (q-1)m},
\label{eq alternative gen matrix linRS}
\end{equation}
where $ G_{k,i} \in \mathbb{F}_{q^m}^{k \times m} $, for $ i = 1,2, \ldots, q-1 $. We may choose such nested linearized Reed-Solomon codes and nested generator matrices by using those in (\ref{eq gen matrix lin RS codes}) or (\ref{eq systematic generator linRS}), for instance. The largest matrix in (\ref{eq alternative gen matrix linRS}), i.e., for $ k = (q-1)m $, can be precomputed and stored for ease of future updates.

Fix an initial number of local groups $ 1 \leq g \leq q-1 $, initial localities $ 1 \leq r_i \leq m $, for $ i = 1,2, \ldots ,g $, and an initial file size $ 1 \leq k \leq \sum_{i=1}^g r_i $. The initial outer code is $ \mathcal{C}_{out} \subseteq \mathbb{F}_{q^m}^N $, $ N = \sum_{i=1}^g r_i $, with generator matrix $ G_k^{in} \in \mathbb{F}_{q^m}^{k \times N} $ obtained by taking the first $ r_i $ columns from $ G_{k,i} $, for $ i = 1,2, \ldots, g $.

Fix $ i = 1,2, \ldots, g $. To go from $ r_i $ to $ r^\prime_i $, we do the following. Let $ \mathbf{c}_{out} \in \mathcal{C}_{out} $ be the outer codeword encoding the file $ \mathbf{f} \in \mathbb{F}_{q^m}^k $. First, decode the $ i $th local group (this has complexity $ \mathcal{O}(r_i^3) $ over $ \mathbb{F}_{q_i} $ in general) to obtain $ \mathbf{c}_{out}^{(i)} \subseteq \mathbb{F}_{q^m}^{r_i} $. Next:
\begin{enumerate}
\item
If $ r^\prime_i < r_i $, then delete the last $ r_i - r_i^\prime $ components of $ \mathbf{c}_{out}^{(i)} \subseteq \mathbb{F}_{q^m}^{r_i} $ to obtain $ \widetilde{\mathbf{c}}_{out}^{(i)} \in \mathbb{F}_{q^m}^{r_i^\prime} $.
\item
If $ r^\prime_i > r_i $, then set $ \widetilde{\mathbf{c}}_{out}^{(i)} = (\mathbf{c}_{out}^{(i)}, \mathbf{f} D) \in \mathbb{F}_{q^m}^{r_i^\prime} $, where $ D $ is the matrix formed by the columns in $ G_{k,i} $ indexed by $ r_i+1 $, $ r_i+2 $, $ \ldots $, $ r_i^\prime $. This has complexity $ \mathcal{O}(k(r_i^\prime - r_i)) $ over $ \mathbb{F}_{q^m} $ in general.
\end{enumerate}
Finally, encode $ \widetilde{\mathbf{c}}_{out}^{(i)} $ using the generator matrix of the new $ r_i^\prime $-dimensional $ i $th local code $ A_i^\prime \in \mathbb{F}_{q_i^\prime}^{r_i^\prime \times n_i^\prime} $, which has complexity $ \mathcal{O}(r_i^{\prime 2}) $ over $ \mathbb{F}_{q_i^\prime} $ in general.

Changing the file size is done as usual with nested linear codes. Assume that $ 1 \leq k < k^\prime \leq \sum_{i=1}^g r_i $, and let $ \mathbf{f} \in \mathbb{F}_{q^m}^k $ and $ \mathbf{f}^\prime = (\mathbf{f}, \mathbf{d}) \in \mathbb{F}_{q^m}^{k^\prime} $ be the initial and final files, respectively. Let $ A_i \in \mathbb{F}_{q_i}^{r_i \times n_i} $ be the generator matrix of the $ i $th local code. If the initial global codeword is $ \mathbf{c}_{glob} \in \mathbb{F}_{q^m}^n $, then the new global codeword is $ \mathbf{c}_{glob}^\prime = \mathbf{c}_{glob} + \mathbf{d} E A \in \mathbb{F}_{q^m}^n $, where $ A = \diag (A_1, A_2, \ldots, A_g) \in \mathbb{F}_q^{N \times n} $, and $ E \in \mathbb{F}_{q^m}^{(k^\prime - k) \times N} $ is formed by the last $ k^\prime - k $ rows in $ G_{k^\prime}^{in} $. Note that going back from $ k^\prime $ to $ k $ is analogous.

We conclude by showing how to add or remove local groups. Let $ 1 \leq g^\prime \leq q-1 $ be the new number of local groups.

Assume first that $ g^\prime < g $ and $ k \leq \sum_{i=1}^{g^\prime} r_i $. In this case, we only need to delete the entire groups indexed by $ i = g^\prime + 1, g^\prime + 2, \ldots, g $, and we are done.

Assume now that $ g^\prime > g $, choose new localities $ 1 \leq r_i \leq m $ and new local generator matrices $ A_i \in \mathbb{F}_{q_i}^{r_i \times n_i} $, for $ i = g+1 $, $ g+2 $, $ \ldots, g^\prime $. Construct $ F \in \mathbb{F}_{q^m}^{k \times N^\prime} $, where $ N^\prime = \sum_{i=g+1}^{g^\prime} r_i $, by taking from $ G_{k,i} $ its first $ r_i $ columns, for $ i = g+1, g+2, \ldots, g^\prime $. If the initial global codeword is $ \mathbf{c}_{glob} \in \mathbb{F}_{q^m}^N $, then the new global codeword is $ (\mathbf{c}_{glob}, \mathbf{f} F A^\prime) \in \mathbb{F}_{q^m}^{N + N^\prime} $, where $ A^\prime = \diag (A_{g+1}, A_{g+2}, \ldots, A_{g^\prime}) \in \mathbb{F}_q^{N^\prime \times n^\prime} $ and $ n^\prime = \sum_{i=g+1}^{g^\prime} n^\prime_i $.

Finally, observe that the three processes described in this subsection (changes of initial localities, file components and number of groups) can be done sequentially in any order, as long as the restriction $ k \leq \sum_{i=1}^g r_i $ is satisfied in all stages.

\section{Comparisons between Different Optimal and/or MR-LRCs for General Parameters} \label{sec comparisons}

In this section, we compare global field sizes of Construction \ref{construction 1} and MR-LRCs from the literatue. We then focus on comparing local field sizes and computational complexities of local and global erasure correction with Construction \ref{construction 1} and codes from the literature that are defined for general parameters. To this end, we will focus on Tamo-Barg codes \cite{tamo-barg}, which are general optimal LRCs with linear field sizes (although not MR); the codes in \cite{gabrys}, which are the previous known MR-LRCs with smallest global fields; and Gabidulin-based LRCs \cite{rawat, kadhe, kim, calis}, which have outer MSRD codes and thus enjoy the same universality, dynamism and local fields as Construction \ref{construction 1}.

We will not study the minimum required per-node storage (i.e., storage complexity), since in all cases the number of bits per symbol is at most linear in $ gr $, which is not large enough to pose problems in large-scale DSSs.

\subsection{Smallest Global Field for given Initial Localities among Linearized Reed-Solomon Codes} \label{subsec smalles field initial localities}

In this subsection, we find that, when using linearized Reed-Solomon codes (Definition \ref{def lin RS codes}) for equal initial localities $ r = r_1 = r_2 = \ldots = r_g $, the smallest field size is achieved by choosing the finest sum-rank length partition, $ N = gr $. In Corollary \ref{cor minimum m for MSRD}, we found that $ r = \max_i r_i $ is the smallest exponent on the global field size, but finer partitions require a larger base, hence so far it is not completely clear that the finest sum-rank partition gives the smallest global field.

\begin{figure} [!h]
\begin{center}
\begin{tabular}{c@{\extracolsep{1cm}}c}
\begin{tikzpicture}

\begin{axis}[
  xmin=0,
  xmax=32,
  ymin=2^0,
  ymax=2^559,
  xtick={1,6,...,31}, 
  ymode=log,
  log basis y={2},
  xlabel={$x$},
  ylabel={$F(x)$}]
  
  \addplot[black, ultra thick, domain=1:31, samples=101,unbounded coords=jump] {max(8,x+1)^(186/x)};
  

\end{axis}
        
\end{tikzpicture}

\end{tabular}
\end{center}

\caption{The function $ F(x) = \max \{ x+1, r+ \delta - 1 \}^{ \lceil gr/x \rceil } $ with logarithmic y-axis (in base $ 2 $), for $ x = 1,2, \ldots, 31 $, where $ g = 31 $, $ r = 6 $ and $ \delta = 3 $. The value $ F(1) = 2^{558} $ indicates the field size required by a Gabidulin code in this case, whereas $ F(31) = 2^{30} $ indicates the field size required by a linearized Reed-Solomon code for the sum-rank length partition $ N = gr = \sum_{i=1}^g r $. }
\label{fig comparison}
\end{figure}
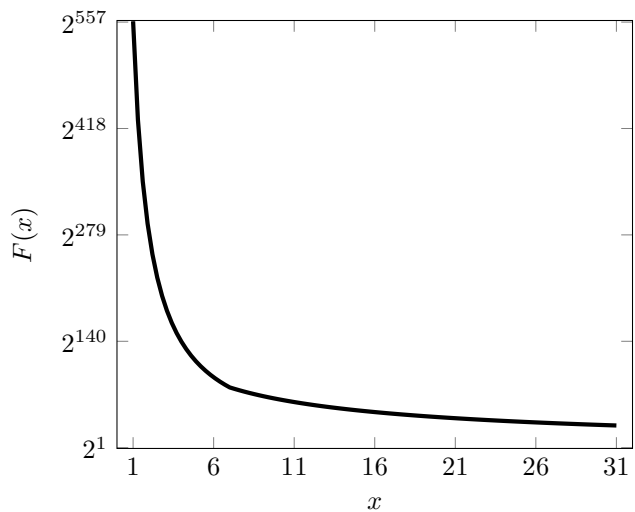

Assume that $ N = gr $ ($ r = r_1 = r_2 = \ldots = r_g $). We argue as follows. By Theorem \ref{th sum-rank partitions}, an MSRD code for a less fine sum-rank length partition than $ N = gr $ is also MSRD for the partition $ N = gr $. Such partitions are of the form $ N = \sum_{i=1}^x r_i^\prime $, for certain $ x = 1,2, \ldots, g $ that acts as the new number of local groups. If the $ r_i^\prime $ are roughly equal, the field size of a linearized Reed-Solomon code with such a sum-rank length partition is roughly $ (x+1)^{\lceil gr/x \rceil} $. Observe that the extremal case $ x=1 $ corresponds to choosing Gabidulin codes. The global field size in Construction \ref{construction 1} would then be 
$$ F(x) = \max \{ x+1, r+ \delta - 1 \}^{ \lceil gr/x \rceil }, $$
since the local MDS codes require field sizes approximately $ r+\delta-1 $. An illustration of this function is given in Fig. \ref{fig comparison} for $ g = 31 $, $ r = 6 $ and $ \delta = 3 $.

First observe that $ F(x) = (r+\delta -1)^{ \lceil gr/x \rceil } $ if $ x \leq r + \delta -2 $, hence the minimum is attained when $ r + \delta - 2 \leq x \leq g $. For these values, we have that $ F(x) = (x+1)^{ \lceil gr/x \rceil } $. Since $ f(x) = \log(x+1) / x $ is decreasing for $ x>0 $, we conclude that 
\begin{equation*}
\begin{split}
(g+1)^r & = F(g) \\
 & = \min \{ \max \{ x+1, r+ \delta - 1 \}^{ \lceil gr/x \rceil } \mid 1 \leq x \leq g \},
\end{split}
\end{equation*}
which is attained by the choice of parameters in Construction \ref{construction 1}, and we are done.

Note that, if $ g \geq 2 $, then any choice $ x \geq 2 $ always decreases the global field size required by Gabidulin codes, without any loss in performance.

\subsection{Comparison with other MR-LRCs from the Literature} \label{subsec other LRC}

In this subsection, we compare Construction \ref{construction 1} with codes from the literature \cite{blaum-RAID, blaum-twoparities, hu, gopalan-MR, calis, gabrys, neri, guruswami-MR}. We will later include the optimal LRCs from \cite{tamo-barg}, although they are not MR-LRCs in general. Throughout this subsection, we fix a dimension $ k $, locality $ r $, local distance $ \delta $, number of local groups $ g $ and number of global parities $ h = gr - k $. The global code length is always $ n = g(r+\delta - 1) $. 

We start by considering families of codes defined for restricted choices of parameters. The work \cite{blaum-RAID} obtains MR-LRCs with linear global field sizes (in the code length $ n $) for $ h = 1 $ and any $ \delta $, and for $ h > 1 $ and $ \delta = 2 $ based on the irreducibility of certain polynomials, which are not known to cover all parameters. A general construction of MR-LRCs for $ \delta = 2 $ with global field sizes of order $ k^h $ is obtained in \cite{gopalan-MR}. MR-LRCs with smaller global field sizes for $ h = 2 $ and $ g=2 $ are obtained in \cite{blaum-twoparities} and \cite{hu}, respectively. 

The first family of MR-LRCs known to cover all parameters was given in \cite{calis} (equivalent to \cite{rawat}). This construction corresponds to our Construction \ref{construction 1} using Gabidulin codes as outer codes, hence it has global field sizes of order $ (r+\delta-1)^{gr} $. The work \cite{gabrys} obtains MR-LRCs with global field sizes of order at least $ g^h = g^{gr-k} $ for general parameters. Observe that, in the PMDS literature, local groups (which are disjoint) are arranged in $ (m \times n) $-array form, as in Figs. \ref{fig LRCs equal} and \ref{fig partitioning localities}. The translation from standard LRC notation to standard PMDS notation is $ m := g $, $ n := r+\delta-1 $, $ s := h = gr-k $, $ r := \delta $. The field sizes in \cite{gabrys} are of orders $ n(mn)^{(r+1)s - 1} $ and $ \max \{ m, n^{r+s} \}^s  $, which in our notation both are at least $ g^h = g^{gr-k} $. This global field size was the smallest so far for general parameters. Recently, the work \cite{neri} obtains MR-LRCs with global field sizes of order $ r^{g(r-\delta + 1)} $ if $ r \geq \delta $. However, in practical scenarios we have that $ g \geq r $, thus $ r^{g(r-\delta + 1)} \gg g^r $ (the size obtained in Construction \ref{construction 1}). Also recently, the work \cite{guruswami-MR} obtains MR-LRCs with global field sizes of order $ n^{\varepsilon h} $, for $ \varepsilon >0 $ satisfying $ h = \Omega(n^{1-\varepsilon}) $ and $ r \ll \varepsilon \log (n) $. However, for such a parameter regime, it holds again that $ n^{\varepsilon h} \gg g^{\varepsilon \log(n)} \gg g^r $.

In conclusion, among families of MR-LRCs that cover general parameters, only those in \cite{gabrys} may have smaller global fields than our Construction \ref{construction 1}. We devote the rest of the subsection to compare, beyond global field sizes, our Construction \ref{construction 1} with the codes in \cite{gabrys}, Tamo-Barg codes \cite{tamo-barg} and MR-LRCs based on Gabidulin codes \cite{rawat, kadhe, kim, calis}. A summary is provided in Table \ref{table comparison}. Note that the global field size of the MR-LRCs from \cite{gabrys} is larger than $ g^h $ for most parameters, but we consider $ g^h $ for simplicity.


\begin{table*}[!t]
\caption{Construction \ref{construction 1} and LRCs from the literature for any dimension $ k $, locality $ r $, number of local groups $ g $ and local distance $ \delta $}
\label{table comparison}
\centering
\begin{tabular}{c||c|c|c|c|c|c}
\hline
&&&&&\\[-1.5em]
LRC family & MR & Global field ($ \approx $) & Local fields ($ \approx $) & Unequal $ (r_i,\delta_i)_{i=1}^g $ & Universal \& Dynamic & Hierarchical  \\[1.6pt]
\hline\hline
&&&&&\\[-1.5em]
Construction \ref{construction 1} & Yes & $ (g+1)^r $ & $ r+\delta-1 $ & Any up to $ r $ & Yes & Yes \\[1.6pt]
\hline
&&&&&\\[-1.5em]
Tamo-Barg \cite{tamo-barg} & No & $ (r+\delta-1)g $ & $ (r+\delta-1) g $ & Unknown & No & Unknown \\[1.6pt]
\hline
&&&&&\\[-1.5em]
Gabrys \textit{et al.} \cite{gabrys} & Yes & $ g^h $, $ h = gr-k $ & $ r+\delta-1 $ & Unknown & No & Unknown \\[1.6pt]
\hline
&&&&&\\[-1.5em]
\cite{rawat, kadhe, kim, calis} & Yes & $ (r+\delta-1)^{gr} $ & $ r+\delta-1 $ & Any choice & Yes & Yes \\[1.6pt]
\hline 
\end{tabular}
\end{table*}


The global field sizes required by Construction \ref{construction 1}, Tamo-Barg codes \cite{tamo-barg}, the codes by Gabrys \textit{et al.} \cite{gabrys} and Gabidulin-based LRCs \cite{rawat, kadhe, kim, calis} are in general,
\begin{equation*}
\begin{split}
q_{C1} & \geq (g+1)^r, \\
q_{TB} & \geq (r + \delta - 1)g, \\
q_{GYBS} & \geq g^{gr - k}, \\
q_{Gab} & \geq (r + \delta - 1)^{gr},
\end{split}
\end{equation*}
respectively, where always $ (r + \delta - 1)^{gr} > (g+1)^r $ as shown in the previous subsection. On the other hand, the minimum field sizes for the $ i $th local code $ \mathcal{C}_{loc}^{(i)} $ are, in general,
\begin{equation*}
\begin{split}
q_{C1}^{loc} & \geq r + \delta - 1, \\
q_{TB}^{loc} & \geq (r + \delta - 1)g, \\
q_{GYBS}^{loc} & \geq r + \delta - 1, \\
q_{Gab}^{loc} & \geq r + \delta - 1, 
\end{split}
\end{equation*}
respectively, where global field sizes must be powers of such local field sizes. See Example \ref{example intro} in Subsection \ref{subsec partitioning initial localities}.

Local repair (LR) by each local code requires the following number of operations over $ \mathbb{F}_2 $ per each block of $ \log(q_{loc}) $ bits, for the corresponding local field size $ q_{loc} $:
\begin{equation*}
\begin{split}
{\rm LR}_{C1} & = \mathcal{O}(r^2 \log(r)^2), \\
{\rm LR}_{TB} & = \mathcal{O}(r^2 \log(g)^2), \\
{\rm LR}_{GYBS} & = \mathcal{O}(r^2 \log(r)^2), \\
{\rm LR}_{Gab} & = \mathcal{O}(r^2 \log(r)^2), 
\end{split}
\end{equation*}
respectively. Here we assume that local decoding algorithms of quadratic complexity exist by Newton-type interpolation. To count the number of operations over $ \mathbb{F}_2 $, we are also assuming that a multiplication in a field $ \mathbb{F} $ of characteristic $ 2 $ costs about $ (\log |\mathbb{F}|)^2 $ operations in $ \mathbb{F}_2 $. 

Finally, global decoding (GD) requires the following number of operations over $ \mathbb{F}_2 $ per each block of $ \log(q_{glob}) $ bits, for the corresponding global field size $ q_{glob} $:
\begin{equation*}
\begin{split}
{\rm GD}_{C1} & = \mathcal{O}(r^4 g^2 \log(g)^2), \\
{\rm GD}_{TB} & = \mathcal{O}(r^2 \log(r)^2 g^2 \log(g)^2), \\
{\rm GD}_{GYBS} & = \mathcal{O}(r^2 (gr-k)^2 g^2 \log(g)^2), \\
{\rm GD}_{Gab} & = \mathcal{O}(r^4 \log(r)^2 g^4),
\end{split}
\end{equation*}
respectively. Here, we assume again that we may apply quadratic-complexity decoding algorithms via Newton-type interpolation (see \cite[Sec. V]{secure-multishot} and \cite[App. B]{secure-multishot} for linearized Reed-Solomon codes) and that a multiplication in $ \mathbb{F} $ costs roughly $ (\log |\mathbb{F}|)^2 $ operations in $ \mathbb{F}_2 $.  

In conclusion, local repair is similar with Construction \ref{construction 1}, \cite{gabrys} and Gabidulin-based LRCs, and is more efficient than for Tamo-Barg codes \cite{tamo-barg}, whereas these latter codes have more efficient global decoding. Gabidulin-based LRCs never achieve more efficient global decoding than Construction \ref{construction 1}. Finally, Construction \ref{construction 1} has more efficient global decoding than \cite{gabrys} whenever $ r \leq h = gr - k $.

Assume now that $ g \gg r $ or that $ r $ is constant, therefore $ g = \Theta(n) $. Then Gabidulin-based LRCs' global decoding complexity is of order $ \mathcal{O}(n^4) $ over $ \mathbb{F}_2 $, and that of the codes in \cite{gabrys} is of order $ \mathcal{O}(h^2 n^2 \log(n)^2) $ over $ \mathbb{F}_2 $. Meanwhile, Construction \ref{construction 1} and Tamo-Barg's global decoding complexities are comparable and of order $ \mathcal{O}(n^2 \log(n)^2) $ over $ \mathbb{F}_2 $, the same as quadratic decoding of Reed-Solomon codes with $ r=1 $, length $ n $, and local replication. See also Example \ref{example second intro} below.

\begin{example} \label{example second intro}
Fix $ r = 9 $, $ \delta = 2 $ and $ h = g $ (thus total length $ n = 10g $). An MR-LRC with such parameters has a local redundancy of $ 10 \% $ (because $ \delta - 1 = (1/10)(r+\delta-1) $) and a global redundancy of $ 10 \% $ (because $ h = (1/10)n $), hence a total redundancy of $ 20 \% $ (note that $ k = gr - h = (8/10)n $), while being able to correct a fraction $ n/10 $ of extra global erasures compared to the Cartesian product of $ g $ $ (r+\delta-1,r) $ MDS codes. The global field sizes in such a case would be roughly $ g^9 = (n/10)^9 $ (polynomial) for our construction, and $ g^h = (n/10)^{n/10} $ (exponential) for the construction in \cite{gabrys}. As argued above, global decoding when keeping $ r $ constant would have complexity of $ \mathcal{O}(n^2 \log(n)^2) $ operations in $ \mathbb{F}_2 $ for our construction, Tamo-Barg codes \cite{tamo-barg} and Reed-Solomon codes \cite{reed-solomon} with local replication. In this case, local repair in our construction is simple XORing, whose complexity does not grow, in contrast to Tamo-Barg codes. 
\end{example}

\section{Further Field Size Reductions: Subextension Subcodes and Sum-rank Alternant Codes} \label{sec further field reductions}

In this section, we introduce the concept of \textit{subextension subcode} of a sum-rank code, which plays the same role as subfield subcodes for the Hamming metric. When applied to linearized Reed-Solomon codes (Definition \ref{def lin RS codes}), we obtain \textit{sum-rank alternant codes}, which have not been considered yet, to the best of our knowledge. 

We will give an estimation on their minimum sum-rank distance and dimension, analogous to the classical estimations for alternant codes \cite{delsarte}. We conclude by analyzing their performance as universal LRCs. Since they can be used as outer codes with the architecture in (\ref{eq architecture arbitrary local codes}), all results in this paper hold, except that recoverability is no longer maximal. As was the case for linearized Reed-Solomon codes, by setting $ m = r_1 = r_2 = \ldots = r_g = 1 $ in this section, we obtain classical alternant codes with arbitrary local replication.

Fix a prime power $ q_0 $, a positive integer $ s $ and $ q = q_0^s $. We also fix a sum-rank length partition $ N = \sum_{i=1}^g r_i $.

\begin{definition}[\textbf{Subextension subcodes}]
Given a code $ \mathcal{C} \subseteq \mathbb{F}_{q^m}^N $, we define its subextension subcode of degree $ m $ over $ \mathbb{F}_{q_0} $ as the subfield subcode 
\begin{equation}
\mathcal{C}_{q_0,m} = \mathcal{C}|_{\mathbb{F}_{q_0^m}} = \mathcal{C} \cap \mathbb{F}_{q_0^m}^N \subseteq \mathbb{F}_{q_0^m}^N.
\end{equation} 
\end{definition}

Denote now by $ \wt_{SR}^q $ and $ \wt_{SR}^{q_0} $ the sum-rank weights in $ \mathbb{F}_{q^m}^N $ and $ \mathbb{F}_{q_0^m}^N $ over $ \mathbb{F}_q $ and $ \mathbb{F}_{q_0} $, respectively. Similarly for distances $ \dd_{SR}^q $ and $ \dd_{SR}^{q_0} $, respectively.

The crucial fact about subextension subcodes is that they inherit their minimum sum-rank distance from the original code. The case $ m=1 $ recovers the well-known fact on the minimum Hamming distance of subfield subcodes.

\begin{theorem} \label{th min sum-rank distance subextension subcodes}
For a (linear or non-linear) code $ \mathcal{C} \subseteq \mathbb{F}_{q^m}^N $, it holds that
\begin{equation}
\dd_{SR}^{q_0}(\mathcal{C}_{q_0,m}) \geq \dd_{SR}^{q}(\mathcal{C}).
\label{eq subext subcode inherits sum-rank distance}
\end{equation}
\end{theorem}
\begin{proof}
Let $ \mathbf{c}, \mathbf{d} \in \mathcal{C}_{q_0,m} $, $ \mathbf{c} \neq \mathbf{d} $, let $ A_i \in \mathbb{F}_{q_0}^{r_i \times r_i} $ be invertible, for $ i = 1,2, \ldots, g $, and define $ A = \diag(A_1, A_2, \ldots, A_g) \in \mathbb{F}_{q_0}^{N \times N} $. Since $ \mathbb{F}_{q_0} \subseteq \mathbb{F}_q $ and $ \mathbf{c}, \mathbf{d} \in \mathcal{C} $, we deduce from Theorem \ref{th sum-rank distance is min among hamming distances} for $ q $ that
$$ \dd_H(\mathbf{c}A, \mathbf{d}A) \geq \dd_{SR}^q(\mathbf{c}A, \mathbf{d}A) = \dd_{SR}^q(\mathbf{c}, \mathbf{d}) \geq \dd_{SR}^q(\mathcal{C}). $$
The result follows now from Theorem \ref{th sum-rank distance is min among hamming distances} for $ q_0 $, after running over all such block-diagonal matrices $ A = $ $ \diag(A_1, A_2,$ $ \ldots, $ $ A_g) $ $ \in \mathbb{F}_{q_0}^{N \times N} $.
\end{proof}

We may now introduce sum-rank alternant codes.

\begin{definition} [\textbf{Sum-rank alternant codes}] \label{def sum-rank alternant}
For a primitive element $ \gamma $ of $ \mathbb{F}_{q^m} $ and a basis $ \mathcal{B} $ of $ \mathbb{F}_{q^m} $ over $ \mathbb{F}_q $, we define the sum-rank alternant code of degree $ m $ over $ \mathbb{F}_{q_0} $, with designed sum-rank distance $ \delta^* $, as the ($ \mathbb{F}_{q_0^m} $-linear) code
$$ \mathcal{C}^{\sigma,q_0,m}_{Alt}(\mathcal{B},\gamma, \delta^*) = (\mathcal{C}^\sigma_{L,\delta^* - 1}(\mathcal{B},\gamma)^\perp)_{q_0,m} \subseteq \mathbb{F}_{q_0^m}^N, $$
where $ \mathcal{C}^\sigma_{L,\delta^* - 1}(\mathcal{B},\gamma) $ is the $ (\delta^* - 1) $-dimensional linearized Reed-Solomon code in Definition \ref{def lin RS codes}.
\end{definition}

We now give estimates on the minimum sum-rank distance and dimension of sum-rank alternant codes.

\begin{corollary} \label{cor estimates sum-rank alternant codes}
The sum-rank alternant code $ \mathcal{C}_{Alt} = \mathcal{C}^{\sigma,q_0,m}_{Alt}(\mathcal{B},\gamma, \delta^*) $ $ \subseteq \mathbb{F}_{q_0^m}^N $ in Definition \ref{def sum-rank alternant} satisfies that:
\begin{enumerate}
\item
$ \dd_{SR}^{q_0}(\mathcal{C}_{Alt}) \geq \delta^* $.
\item
$ \dim(\mathcal{C}_{Alt}) \geq N - s(\delta^* - 1) $, where $ q = q_0^s $.
\end{enumerate}
\end{corollary}
\begin{proof}
First, the dual of a linearized Reed-Solomon code is again a linearized Reed-Solomon code (see \cite[Th. 4]{secure-multishot}), hence is also MSRD. Thus Item 1 follows from Theorem \ref{th min sum-rank distance subextension subcodes}.

Next, since $ q^m = (q_0^m)^s $ by hypothesis, Item 2 follows from Delsarte's lower bound on dimensions of subfield subcodes \cite{delsarte}. 
\end{proof}

Observe that setting $ m=1 $, we recover the Hamming metric, classical Reed-Solomon codes and classical alternant codes. The previous estimations become then the classical ones \cite{delsarte}. 

From the study in Section \ref{sec universal lrc arbitrary local codes}, we deduce the following on sum-rank alternant-based LRCs. 

\begin{theorem}
Let $ \mathcal{C}_{out} = \mathcal{C}^{\sigma,q_0,m}_{Alt}(\mathcal{B},\gamma, \delta^*) $ $ \subseteq \mathbb{F}_{q_0^m}^N $ be the sum-rank alternant code in Definition \ref{def sum-rank alternant}. Fix full-rank matrices $ A_i \in \mathbb{F}_{q_0}^{r_i \times n_i} $, $ 1 \leq r_i \leq n_i $, for $ i = 1,2, \ldots, g $, and define $ A = $ $ \diag(A_1, A_2,$ $ \ldots, $ $ A_g) $ $ \in \mathbb{F}_{q_0}^{N \times n} $ and $ n = n_1 + n_2 + \cdots + n_g $. The global code $ \mathcal{C}_{glob} = \mathcal{C}_{out}A \subseteq \mathbb{F}_{q_0^m}^n $ (see (\ref{eq architecture arbitrary local codes})) is a $ (\Gamma_i, \mathcal{C}_i)_{i=1}^g $-LRC as in Definition \ref{def locally linear LRC}, such that:
\begin{enumerate}
\item
If $ \mathcal{R} \subseteq [n] $ and $ \Rk(A|_\mathcal{R}) \geq N - \delta^* + 1 $, then the erasure pattern $ \mathcal{E} = [n] \setminus \mathcal{R} $ can be corrected by $ \mathcal{C}_{glob} $ for all codewords $ \mathbf{c} \in \mathcal{C}_{glob} $.
\item
It holds that 
$$ e(A, N - \delta^* + 1) \leq \dd_H(\mathcal{C}_{glob}) - 1 \leq e(A,N - s (\delta^*-1)), $$
where $ e(A,k) $ is as in (\ref{eq def of e_max}), for $ k \in \mathbb{N} $.
\end{enumerate}
\end{theorem}
\begin{proof}
Item 1 follows from Corollary \ref{cor sum-rank erasure correction} and $ \dd_{SR}^{q_0}(\mathcal{C}_{out}) \geq \delta^* $ (Corollary \ref{cor estimates sum-rank alternant codes}). The first inequality in Item 2 follows from Item 1 and (\ref{eq def of e_max}). The second inequality in Item 2 follows from $ \dim(\mathcal{C}_{glob}) \geq N - s(\delta^* - 1) $ (Corollary \ref{cor estimates sum-rank alternant codes}), Theorem \ref{th global distance and threshold} and the fact that $ e(A,k) \leq e(A,k^\prime) $ whenever $ k \geq k^\prime $, which follows from (\ref{eq def of e_max}).
\end{proof}

In conclusion, we obtain an exponential reduction in field size, with exponent $ s $, by reducing the entropy of the stored file by at most $ (s - 1)(\delta^* - 1) $. However, this reduction of information rate is only a bound. It would be of interest to find sharper lower bounds on the dimension of subextension subcodes, as done in \cite{stichtenoth} for classical subfield subcodes. Adapting known decoding algorithms of alternant codes is also of interest, as well as finding what type of cyclicity certain sum-rank alternant codes may have.

\section{Conclusion} \label{sec conclusion}

In this work, we have proposed an architecture for LRCs based on those in \cite{rawat, kadhe, kim, calis}, but substituting Gabidulin codes \cite{gabidulin, roth} by general MSRD codes, in particular, by linearized Reed-Solomon codes \cite{linearizedRS}. Construction \ref{construction 1} achieves maximal recoverability and all the flexibility advantages of Gabidulin-based LRCs, but with global field sizes roughly $ g^r $. Such field sizes improve the smallest known global fields of MR-LRCs when $ r \leq h $ \cite{gabrys} for equal localities, and all previous best known MR-LRCs for unequal localities \cite{kadhe, kim}. 

The flexibility features of Construction \ref{construction 1} include being compatible with arbitrary local linear codes (not necessarily MDS) over much smaller local fields, partitioning the initial localities without global recoding, and changing the initial localities, file components and number of local groups, without global recoding. It also enabled us to obtain explicit multi-layer or hierarchical MR-LRCs for any type of hierarchy and any choice of (equal or unequal) localities and local distances.

To further reduce global field sizes, subextension subcodes and sum-rank alternant codes have been introduced. As in the classical case, exponential field size reductions are possible at the cost of reducing the information rate.

\ifCLASSOPTIONcaptionsoff
  \newpage
\fi



\bibliographystyle{IEEEtran}
\end{document}